\documentclass[a4paper,UKenglish,cleveref, autoref, thm-restate]{lipics-v2021}

\title{Subcoloring of (Unit) Disk Graphs} 

\author{Malory {Marin}}{ENS de Lyon, Université Claude Bernard Lyon 1, CNRS, Inria, LIP UMR 5668, Lyon, France}{malory.marin@ens-lyon.fr}{https://orcid.org/0009-0008-8253-2831}{}

\author{Rémi {Watrigant}}{Université Claude Bernard Lyon 1, ENS de Lyon, CNRS, Inria, LIP UMR 5668, Lyon, France}{rémi.watrigant@univ-lyon1.fr}{https://orcid.org/0000-0002-6243-5910}{}
\authorrunning{M. Marin and R. Watrigant} 

\Copyright{Malory Marin and Rémi Watrigant} 



\acknowledgements{We would like to thank Julien Duron, Jean-Florent Raymond, and Sébastien Very for valuable discussions on the subject.}

\hideLIPIcs

\EventEditors{John Q. Open and Joan R. Access}
\EventNoEds{2}
\EventLongTitle{42nd Conference on Very Important Topics (CVIT 2016)}
\EventShortTitle{CVIT 2016}
\EventAcronym{CVIT}
\EventYear{2016}
\EventDate{December 24--27, 2016}
\EventLocation{Little Whinging, United Kingdom}
\EventLogo{}
\SeriesVolume{42}
\ArticleNo{23}

\usepackage{xspace}
\newcommand{\cs}{\chi_{s}}

\nolinenumbers

\usepackage{tikz}
\usetikzlibrary{shapes}
\usetikzlibrary{positioning,calc} 
\usepackage{lipsum}

\begin{document}

\maketitle

\begin{abstract}
A subcoloring of a graph is a partition of its vertex set into subsets (called colors), each inducing a disjoint union of cliques. It is a natural generalization of the classical proper coloring, in which each color must instead induce an independent set.
Similarly to proper coloring, we define the \textit{subchromatic number} of a graph as the minimum integer $k$ such that it admits a subcoloring with $k$ colors, and the corresponding problem \textsc{$k$-Subcoloring} which asks whether a graph has subchromatic number at most $k$. 
In this paper, we initiate the study of the subcoloring of (unit) disk graphs.
One motivation stems from the fact that disk graphs can be seen as a dense generalization of planar graphs where, intuitively, each vertex can be blown into a large clique--much like subcoloring generalizes proper coloring. 
Interestingly, it can be observed that every unit disk graph admits a subcoloring with at most $7$ colors.
We first prove that the subchromatic number can be $3$-approximated in polynomial-time in unit disk graphs. We then present several hardness results for special cases of unit disk graphs which somehow prevents the use of classical approaches for improving this result.
We show in particular that \textsc{$2$-subcoloring} remains NP-hard in triangle-free unit disk graphs, as well as in unit disk graphs representable within a strip of bounded height. We also solve an open question of Broersma, Fomin, Ne\v{s}et\v{r}il, and Woeginger (2002) by proving that \textsc{$3$-Subcoloring} remains NP-hard in co-comparability graphs (which contain unit disk graphs representable within a strip of height $\sqrt{3}/2$).
Finally, we prove that every $n$-vertex disk graph admits a subcoloring with at most $O(\log^3(n))$ colors and present a $O(\log^2(n))$-approximation algorithm for computing the subchromatic number of such graphs. This is achieved by defining a decomposition and a special type of co-comparability disk graph, called $\Delta$-disk graphs, which might be of independent interest.
\end{abstract}

\newpage

\section{Introduction}

\subsection{Definitions and Related Work}

One of the most fundamental topic in algorithmics and graph theory is the coloring problem. A \textit{proper coloring} of a graph is an assignment of colors to the vertices of the graph in such a way that no two adjacent vertices receive the same color, or equivalently, such that each color induces an independent set. The \textit{chromatic number} of a graph $G$, denoted by $\chi(G)$, is the minimum number of colors of a proper coloring of $G$.
In this paper we deal with a natural variant of this number.
  In 1989, Albertson et al. introduced the notion of \textit{subcoloring} of a graph \cite{albertson1989subchromatic}, which is a coloring where each color induces a disjoint union of complete graphs, or equivalently, a graph with no induced path on three vertices. The \textit{subchromatic number} of a graph $G$, denoted by $\cs(G)$, is the minimum number of colors needed to obtain a subcoloring of the graph. By definition, any proper coloring is also a subcoloring, and thus $\cs(G) \leqslant \chi(G)$ holds for any graph $G$. However, the gap between these two parameters may be arbitrarily large, as in the case of complete graphs.
  The subchromatic number is also upper bounded by two other well-known parameters: the clique cover number $\text{cc}(G)$, which is the chromatic number of the complement graph $\overline{G}$; and the co-chromatic number $\text{co}\chi(G)$ which is the minimum number of colors needed to color the graph such that each color class induce either a clique or an independent set.

The \textsc{$k$-Subcoloring} problem, which consists in deciding whether the subchromatic number of a given graph is at most $k$, received significant attention, both from structural and algorithmic aspects. 
Its complexity greatly differs from the one of the classical coloring problem. For example, \textsc{$k$-Subcoloring} is NP-hard already for $k=2$. In addition, \textsc{$2$-Subcoloring} remains NP-hard in various subclasses of planar graphs, namely triangle-free planar graphs with maximum degree~$4$~\cite{fiala,gimbel}, planar graphs with girth $5$~\cite{MoOc15}, and planar comparability graphs with maximum degree $4$~\cite{ochem}. Notice also that any graph with maximum degree $3$ is $2$-subcolorable, following an argument of Erd\H{o}s~\cite{albertson1989subchromatic,Erdos67}.
In 2002, Broersma et al.~\cite{broersma2002} studied the \textsc{$k$-Subcoloring} problem in several subclasses of perfect graphs, in which the chromatic number is well understood. Among others, they showed that the problem is NP-complete in comparability graphs for $k = 2$, and that in the case of interval graphs, the problem can be solved in time $O(kn^{2k+1})$, where $n$ is the number of vertices. Additionally, they proved that every chordal graph admits a solution with a logarithmic number of colors, which in particular implies that in interval graphs, when $k$ is part of the input, the problem is solvable in quasi-polynomial time. 
Following this work, it was established that in chordal graphs, the problem is polynomial-time solvable for $k = 2$~\cite{Stacho08}, but becomes NP-hard for $k = 3$~\cite{StachoPhDthesis}. Additionally, a $3$-approximation algorithm was obtained for computing the subchromatic number of interval graphs~\cite{RaBrSrRa10}. 
Interestingly, two questions posed by Broersma et al. remain open and have a geometric flavor:

\begin{itemize}
    \item What is the complexity of \textsc{$k$-Subcoloring} in interval graphs when $k$ is part of the input? 
    \item What is the complexity of \textsc{$k$-Subcoloring} in co-comparability graphs? (Note that co-comparability graphs are intersection graphs of curves from a line to a parallel line)
\end{itemize}


\begin{figure}[t]
\begin{tikzpicture}[scale=0.9]
\node[draw,rectangle, text width = 3.8cm] (Perfect) at (-2,2.5) {\textbf{Perfect} \\ NP-c for fixed $k\geqslant 2$};
\node[draw, rectangle, text width = 3cm] (Chor) at (2,0.5) {\textbf{Chordal} \\ NP-c for $k\geqslant 3$\\ P for $k=2$} ;
\node[draw, rectangle, fill = green, fill opacity = 0.2, text width = 3.3cm, text opacity = 1] (Co-Comp) at (-2.5,-1) {\textbf{Co-comparability} \\ NP-c for $k\geqslant 3$ \\ Open for $k=2$} ;
\node[draw, rectangle, , fill = green, fill opacity = 0.2, text width = 2.9cm, text opacity = 1] (Delta) at (3.9,-1.3) {\textbf{$\Delta$-disks} \\ $c$-approx for $\cs$} ;
\node[draw, rectangle, text width = 3cm] (Perm) at (-2.5,-4) {\textbf{Permutation} \\ P for fixed $k$ } ;
\node[draw, rectangle, text width = 3.4cm] (Int) at (2,-4) {\textbf{Intervals} \\ P for fixed $k$ \\ $3$-approx. of $\cs$}  ;
\node[draw, rectangle, fill = green, fill opacity = 0.2, text width = 3.3cm, text opacity = 1] (Disk) at (6,2) {\textbf{Disks} \\ NP-c for fixed $k\geqslant 2$} ;
\node[draw, rectangle, text width = 2.8cm] (Planar) at (7.5,-3.5) {\textbf{Planar} \\ NP-c for $k=2,3$} ;
\node[draw, rectangle, fill = green, fill opacity = 0.2, text width = 2.6cm, text opacity = 1] (UDisk) at (8.5,-0.4) {\textbf{Unit disks} \\ NP-c for $k=2$} ;

\draw[->] (Perfect) -- (Co-Comp);
\draw[->] (Perfect) -- (Chor);
\draw[->] (Chor) -- (Int);
\draw[->] (Co-Comp) -- (Perm);
\draw[->] (Co-Comp) -- (Delta);
\draw[->] (Disk) -- (Delta);
\draw[->] (Delta) -- (Int);
\draw[->] (Disk) -- (Planar);
\draw[->] (Disk) -- (UDisk);

\draw[dashed, red, line width = 1.3 ] (-0.3,-5) .. controls (-0.3,3.5) .. (10,3);

\draw[dashed, red, line width = 1.3  ] (5.5,-5) .. controls (6,1) .. (10,1);

\node[draw = none] () at (9,3.5) {\color{red}  Polynomial};
\node[draw = none] () at (9,2) {\color{red}  Polylog};
\node[draw = none] () at (9,0.5) {\color{red}  Constant};

\end{tikzpicture}
\caption{Complexity of the \textsc{$k$-Subcoloring} problem on various class of graphs, with arrows representing the inclusion relationships between these classes. The classes in green represent the specific classes for which we have made contributions. The red dashed lines correspond to extremal values of the subchromatic number as a function of the number of vertices.}\label{ComplexitySchema}
\end{figure}
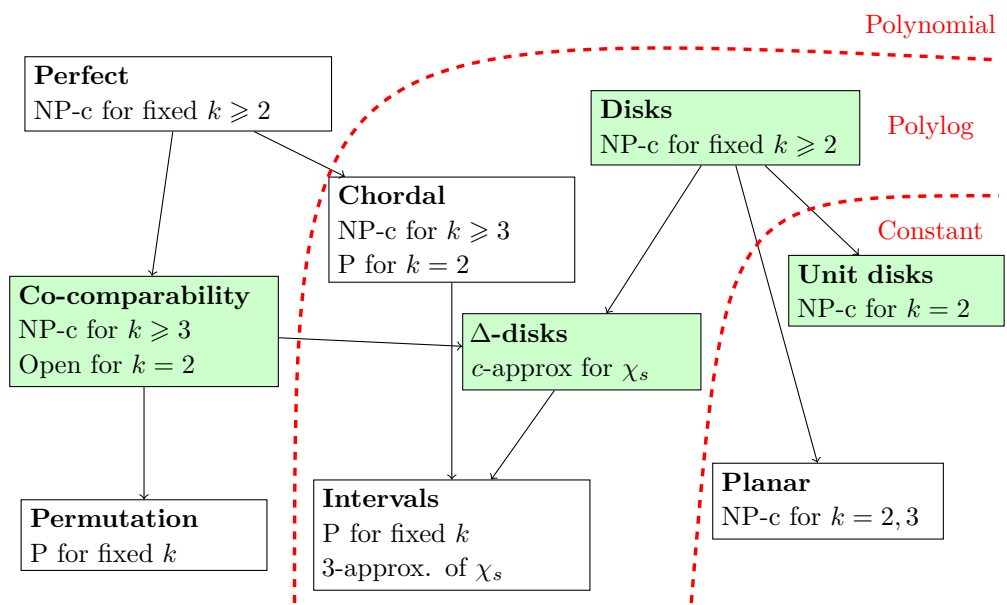

In this paper, we initiate the study of the \textsc{$k$-Subcoloring} problem in disk graphs, that is, in intersection graphs of a set of disks drawn on the plane. A celebrated result of Koebe states that every planar graph is also a disk graph~\cite{Koebe1936}. The converse is not true, as disks graph can contain large cliques. However, one might think that the presence of large cliques essentially captures the only difference between these two classes, as triangle-free disk graphs are indeed planar graphs~\cite{Breu1996}. 
This observation greatly motivates the study of the subchromatic number of disk graphs since, analogously, this parameter can be seen as a generalization of the chromatic number in which each vertex is blown into a clique.
Despite its seemingly nice geometrical properties, disk graphs still posses a complex combinatorial structure, and the complexity of many fundamental problems are still poorly understood in this graph class. For instance, it is still open whether \textsc{Maximum Clique} is polynomial-time solvable in disk graphs~\cite{BoBoBoChGiKiRzSiTh21}.
As we will see, the subcoloring problem also suffers from this phenomenon.
One of the reason is perhaps that disk graphs also capture interval graphs, containing themselves graphs of arbitrarily large subchromatic number~\cite{broersma2002}.
A natural restriction of disk graphs is the class of \textit{unit} disk graphs, in which all disks in the representation have diameter $1$. Although unit disk graph might still contain large cliques, their simpler structure brings them a little closer to planar graphs. For instance, several algorithmic techniques originally designed for planar graphs (and minor-free classes), such as Baker's technique, can sometimes be applied to unit disk graphs. To do so, one generally has to partition the plane into \textit{strips} of constant height, and then solve the considered problem optimally in each strip, by either using dynamic programming techniques or other structural argument. For instance, it can be shown that unit disk graphs representable within a strip of height $\sqrt{3}/2$ are co-comparability graphs (which are perfect graphs).

\paragraph*{Practical motivation}

Geometric graphs serve as a natural model for wireless networks. Two vertices (i.e., access points of such a network) can detect each other's transmission within a certain range.
This is a classical motivation for studying proper coloring of (unit) disk graphs: assigning channels to access points such that all access points within a channel can communicate at the same time without interference. However in practice networks tend to be too dense, and thus the chromatic number is too large compared to the number of available channels. As a consequence, no matter the channel allocation, it is unavoidable to have interference within channels. A formal study on the structure of wireless networks under saturation conditions has been developed, and it was observed that the throughput of any access point is (somewhat surprisingly) proportional to the number of MIS (Maximum Independent Sets) the vertex belongs to, within its channel~\cite{csma-capacity,BeBuFeWa23}. In particular, if a vertex does not belong to any MIS, its throughput will be very close to zero. This is where the subcoloring comes into play. In a disjoint union of cliques, each vertex is indeed in a MIS, and this structure offers great stability due to its hereditary nature: any access points can fail or be removed without significantly changing the aforementioned property. It can even be proved that this is the larger class of graphs with both properties.

\subsection{Contribution and organization}

The next section presents the technical preliminaries essential for the rest of the paper.

In Section~\ref{sec:positiveunit}, we start by examining the \textsc{$k$-Subcoloring} problem on unit disk graphs. We first establish that the subchromatic number of any unit disk graph is at most $7$, which immediately leads to a simple $\frac{7}{2}$-approximation algorithm. We then refine this result, providing a $3$-approximation algorithm, which matches the best-known approximation ratio for the classical \textsc{Coloring} problem on unit disk graphs.

In Section~\ref{sec:hardness}, we demonstrate the challenges in further improving the approximation algorithm by proving several hardness results. Specifically:

\begin{theorem}\label{thm:hardness}
The \textsc{$k$-Subcoloring} problem is NP-complete on:
\begin{enumerate}
    \item Triangle-free unit disk graphs for $k = 2$.
    \item Unit disk graphs with bounded height for $k = 2$.
    \item Co-comparability graphs for any $k \geq 3$.
\end{enumerate}
\end{theorem}

The third result resolves an open question posed by Broersma et al. on the complexity of \textsc{$k$-Subcoloring} in co-comparability graphs.

In Section~\ref{sec:PositiveDisk}, we extend our study to general disk graphs. The methods used to achieve a $3$-approximation for interval graphs cannot be directly generalized to disk graphs. To address this, we introduce a new subclass of disk graphs called $\Delta$-disk graphs, which (i) generalize interval graphs, (ii) belong to the class of co-comparability graphs, and (iii) serve as building blocks for decomposing general disk graphs. The structured properties of $\Delta$-disk graphs enable us to derive a logarithmic upper bound and a constant-factor approximation algorithm for their subchromatic number. These results lead to the two main contributions of this section:

\begin{theorem}\label{thm:UpperBoundDiskGraph}
For any $n$-vertex disk graph $G$, the subchromatic number $\chi_s(G)$ satisfies $\chi_s(G) = O(\log^3 n)$.
\end{theorem}

\begin{theorem}\label{thm:ApproxDiskGraph}
There exists an $O(\log^2 n)$-approximation algorithm for computing the subchromatic number of an $n$-vertex disk graph.
\end{theorem}

We believe that $\Delta$-disk graphs and the decomposition technique employed may be of independent interest.

\section{Preliminaries}\label{sec:preliminaries}
Most definitions used in this paper are standard. We nevertheless recall the ones that are most used.

\paragraph*{Graph notations} 
Let $G$ be a simple graph. We denote by $V(G)$ and $E(G)$ the set of vertices and the set of edges of $G$, respectively. When there is no ambiguity, we denote by $n$ the number of vertices of $G$, and by $m$ the number of edges of $G$. An \emph{independent set} of $G$ is a set of pairwise non-adjacent vertices, and we denote by $\alpha(G)$ the \emph{independence number} of $G$, i.e., the size of a maximum independent set. Similarly, a \emph{clique} in $G$ is a set of pairwise adjacent vertices, and we denote by $\omega(G)$ the size of a maximum clique.

Given a set $R \subseteq V(G)$, we use $G[R]$ to denote the subgraph induced by $R$, and $G - R$ to denote the graph induced by $V(G) \setminus R$. For a vertex $v \in V(G)$, we denote by $N(v)$ the \emph{open neighborhood} of $v$, that is, $N(v) = \{u \in V(G) \mid uv \in E(G)\}$, and by $N[v]$ its \emph{closed neighborhood}, defined as $N[v] = N(v) \cup \{v\}$.

Given an integer $t \geqslant 1$, we denote by $K_t$ the clique on $t$ vertices. A graph $G$ is said to be \emph{$K_t$-free} if it does not contain $K_t$ as an induced subgraph. In the special case of $t = 3$, we use the term \emph{triangle-free graph} instead of \emph{$K_3$-free graph}.

A graph $G$ is said to be a \emph{disjoint union of cliques} if there exists a partition of $V(G)$ into $V_1, \ldots, V_\ell$ such that, for all $1 \leqslant i \leqslant \ell$, $G[V_i]$ is a clique, and there are no edge between vertices in different cliques. Equivalently, $G$ is a disjoint union of cliques if it does not contain an induced path on three vertices (i.e., it is $P_3$-free).

A graph is a \emph{co-comparability} if it is the intersection graph of curves from a line to a parallel line. Equivalently, a co-comparability graph is a graph that connects pairs of elements that are not comparable to each other in a partial order. 

A \emph{hypergraph} $H = (V, \mathcal{E})$ consists of a set of vertices $V$ and a set of edges $\mathcal{E}$, where each edge is a subset of $V$. If all edges have exactly $k$ vertices, the hypergraph is called \emph{$k$-uniform}. For example, in a $3$-uniform hypergraph, every edge connects exactly three vertices. Similarly to graphs, we denote by $V(H)$ the set of vertices of $H$ and by $\mathcal{E}(H)$ its set of hyperedges.

\paragraph*{Geometric notations} 

Given two points $A, B \in \mathbb{R}^2$, we denote by $d(A, B)$ their Euclidean distance. A disk of radius $r \in \mathbb{R}$ and center $A \in \mathbb{R}^2$ is the set $D(A, r) = \{B \in \mathbb{R}^2 \mid d(A, B) \leqslant r\}$. The diameter of such a disk is $2r$, and if a disk has diameter $1$, it is called a \emph{unit disk}. 

Note that the intersection of two disks $D(A, r_A)$ and $D(B, r_B)$ is non-empty if and only if $d(A, B) \leqslant r_A + r_B$. A \emph{(unit) disk graph} is the intersection graph of (unit) disks in the plane. The set of disks $\mathcal{D}_G = \{D_v = D((x_v, y_v), r_v) \mid v \in V(G)\}$ of a given disk graph $G$ is called a \emph{disk representation} of $G$. For a real number $h > 0$, a unit disk graph $G$ is said to have \emph{height} $h$ if there exists a disk representation of $G$ such that the center of every disk lies within the region $\mathbb{R} \times [0, h]$.

An \emph{interval graph}  is the intersection graph on intervals on the real line. The set of intervals $\mathcal{I}_G = \{(\ell_v,r_v)\}_{v\in V(G)}$ is called an \emph{interval representation} of $G$.

\paragraph*{Approximation scheme} A $\lambda$-approximation algorithm is an algorithm that runs in polynomial time with respect to the instance size and returns a solution that is guaranteed, in the worst case, to be within a factor of $\lambda$ from the optimal solution.

\section{First results on unit disk graphs}\label{sec:positiveunit}

As a warm-up, let us first consider the case of unit disk graphs.

 A simple yet important result is that the subchromatic number of a unit disk graph is bounded by a constant. To provide some intuition for this result, imagine partitioning the plane into a grid, where each cell has a diagonal of length 1. All vertices within a cell form a clique, while the cells can be colored using a constant number of colors, ensuring that no two cells within a distance less than 1 share the same color. To optimize this constant, we consider a hexagonal tiling of the plane and the subcoloring induced by this tiling, usually known as \textit{Isbell's coloring}. Notice that in order to turn the following property into a polynomial-time algorithm, we need the unit disk embedding of the input graph.

\begin{observation}\label{thm:unitdisk7}
Every unit disk graphs $G$ satisfies $\cs(G) \leq 7$.
\end{observation}

\begin{proof}
Let $G=(V,E)$ be a unit disk graph. We consider a hexagonal tiling of the plane, which can be represented as a function $f : \mathbb{Z}^2 \rightarrow \mathcal{H}$, where each element of $\mathcal{H}$ corresponds to a hexagon of diameter $1$ in $\mathbb{R}^2$. We color each hexagon of $f(\mathbb{Z}^2)$ such that for any $(i,j) \in \mathbb{Z}^2$, $f(i,j+1)=f(i,j)+1 \mod 7$, and $f(i+1,j)=f(i,j)+5 \mod 7$. In other words, if a hexagon has color $c$, then the one on its right will have color $c+1 \mod 7$, and the two below it will have colors $c+4 \mod 7$ and $c+5 \mod 7$ respectively, as described in Figure~\ref{fig:Hexagon}.

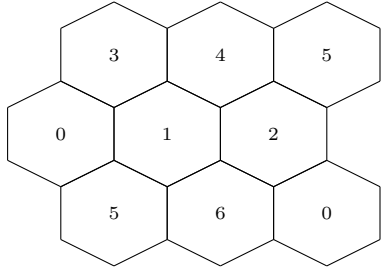
\begin{figure}[h]
\begin{center}
\begin{tikzpicture}[scale=0.7]

\foreach \j in {0}{
	\foreach \i in {2,4,6}{
		\draw (\i-1, \j-0.5) -- (\i-1, \j+0.5) ;
		\draw (\i+1, \j-0.5) -- (\i+1, \j+0.5) ;
		\draw (\i, \j-1) -- (\i-1, \j-0.5);
		\draw (\i, \j-1) -- (\i+1, \j-0.5);
		\draw (\i, \j+1) -- (\i-1, \j+0.5);
		\draw (\i, \j+1) -- (\i+1, \j+0.5);
	};
}

\foreach \j in {1.5}{
	\foreach \i in {1,3,5}{
		\draw (\i-1, \j-0.5) -- (\i-1, \j+0.5) ;
		\draw (\i+1, \j-0.5) -- (\i+1, \j+0.5) ;
		\draw (\i, \j-1) -- (\i-1, \j-0.5);
		\draw (\i, \j-1) -- (\i+1, \j-0.5);
		\draw (\i, \j+1) -- (\i-1, \j+0.5);
		\draw (\i, \j+1) -- (\i+1, \j+0.5);
	};
}
\foreach \j in {3}{
	\foreach \i in {2,4,6}{
		\draw (\i-1, \j-0.5) -- (\i-1, \j+0.5) ;
		\draw (\i+1, \j-0.5) -- (\i+1, \j+0.5) ;
		\draw (\i, \j-1) -- (\i-1, \j-0.5);
		\draw (\i, \j-1) -- (\i+1, \j-0.5);
		\draw (\i, \j+1) -- (\i-1, \j+0.5);
		\draw (\i, \j+1) -- (\i+1, \j+0.5);
	};
}

\node[] (0) at (2,0) {\scriptsize $5$};
\node[] (1) at (4,0) {\scriptsize $6$};
\node[] (2) at (6,0) {\scriptsize $0$};
\node[] (3) at (1,1.5) {\scriptsize $0$};
\node[] (4) at (3,1.5) {\scriptsize $1$};
\node[] (5) at (5,1.5) {\scriptsize $2$};
\node[] (6) at (2,3) {\scriptsize $3$};
\node[] (7) at (4,3) {\scriptsize $4$};
\node[] (8) at (6,3) {\scriptsize $5$};

\end{tikzpicture}
\end{center}
\caption{Isbell's coloring of the plane.}\label{fig:Hexagon}
\end{figure}
Since each hexagon has a diameter of $1$, all disks whose centers belong to a same hexagon form a clique. Then, if two hexagons have the same color, we distinguish two cases. If they are in the same row, there are $6$ hexagons between them, and thus their distance is at least $6 \times \sqrt{3}/2 > 1$. Then, if they do not belong to the same row, their minimum distance is at least $\sqrt{7}/2 > 1$ (see the two cells with color $0$ on the figure and the two of color $5$).
\end{proof}

An algorithmic consequence of Theorem~\ref{thm:unitdisk7} is the existence of a straightforward polynomial-time $7/2$-approximation algorithm for computing the subchromatic number, since deciding whether a graph is a disjoint union of cliques (hence deciding if the subchromatic number is $1$) can be done in linear time.

\begin{corollary}
There exists a $7/2$-approximation for the subchromatic number on unit disk graphs, running in linear time, given the unit disk representation.
\end{corollary}

Here we improve the approximation ratio by presenting a $3$-approximation algorithm. A key step involves partitioning the unit disk graph so that each part induces a disjoint union of graphs with a maximum independent set of bounded size. Then, we conclude by using the fact that in general graphs, \textsc{$2$-Subcoloring} can be solved in FPT time when parameterized by the total number of connected components of each color~\cite{kanj2018parameterized}. Naturally, this parameter is upper bounded by the independence number of the input graph.

\begin{theorem}
There exists a $3$-approximation for the subchromatic number in unit disk graphs running in time $O(n\cdot m)$, given the unit disk representation.
\end{theorem}

\begin{proof}
The algorithm is based on a partition of the vertices of the unit disk graph into three sets $V_1, V_2, V_3$ such that for any $i\in \{1,2,3\}$, $G[V_i]$ is a disjoint union of graphs with constant independence number, allowing us to decide in polynomial-time whether it admits a $2$-subcoloring. If so, we will thus be able to output a $6$-subcoloring of $G$. Otherwise, if some $G[V_i]$ cannot be $2$-subcolored, we return a $7$-subcoloring with Observation~\ref{thm:unitdisk7}. Let 
\begin{align*}
R_0 &= \{(x,y) \in \mathbb{R}^2 \mid \lfloor y \rfloor = 0 \mod 2 \wedge \lfloor x\rfloor  \neq 0 \mod 4\} \\
R_1 &= \{(x,y) \in \mathbb{R}^2 \mid \lfloor y \rfloor = 1 \mod 2 \wedge \lfloor x\rfloor  \neq 3 \mod 4\} \\
R_2 &= \mathbb{R}^2 \backslash (V_0\cup V_1)
\end{align*}
Note that $(R_0, R_1, R_2)$ forms a partition of the plane, and each $R_i$ consists of rectangles which are are distance at least $1$ from each other. The partition is depicted in Figure \ref{3ApproxPartition}.

\begin{figure}
\centering
\begin{tikzpicture}

\foreach \i in {0,2}{
	\draw[draw = none, fill = lipicsYellow, fill opacity = 0.4] (0,\i) rectangle (3, \i+1) ;
	\draw[draw = none, fill = lipicsGray, fill opacity = 0.4] (3,\i) rectangle (4, \i+1) ;
	\draw[draw = none, fill = lipicsYellow, fill opacity = 0.4] (4,\i) rectangle (5, \i+1) ;
	\draw[draw = none, fill = lipicsGray, fill opacity = 0.4] (-1,\i) rectangle (0, \i+1) ;
	\draw[draw = none, fill = lipicsYellow, fill opacity = 0.4] (-4,\i) rectangle (-1, \i+1) ;
	\draw[draw = none, fill = lipicsGray, fill opacity = 0.4] (-5,\i) rectangle (-4, \i+1) ;
}

\foreach \i in {1,3} {
    \draw[draw=none, fill=purple, fill opacity=0.4] (-2,\i) rectangle (1, \i+1);
    \draw[draw=none, fill=lipicsGray, fill opacity=0.4] (1,\i) rectangle (2, \i+1);
    \draw[draw=none, fill=purple, fill opacity=0.4] (2,\i) rectangle (5, \i+1);
    \draw[draw=none, fill=lipicsGray, fill opacity=0.4] (-3,\i) rectangle (-2, \i+1);
    \draw[draw=none, fill=purple, fill opacity=0.4] (-5,\i) rectangle (-3, \i+1);
}

\foreach \x/\y in {0.245/2.401, -1.957/2.069, 0.478/1.758, 1.336/0.175, 1.033/0.989, 2.618/2.81, 0.648/3.351, -1.463/3.705, -4.264/3.012, -1.822/2.488}{
	\draw (\x,\y) circle (1/2) node{$\bullet$};
}

\end{tikzpicture}
\caption{The $3$-partition of the plane to obtain the $3$-approximation algorithm. The yellow represents $R_0$, the purple $R_1$ and the gray $R_2$.}\label{3ApproxPartition}
\end{figure}
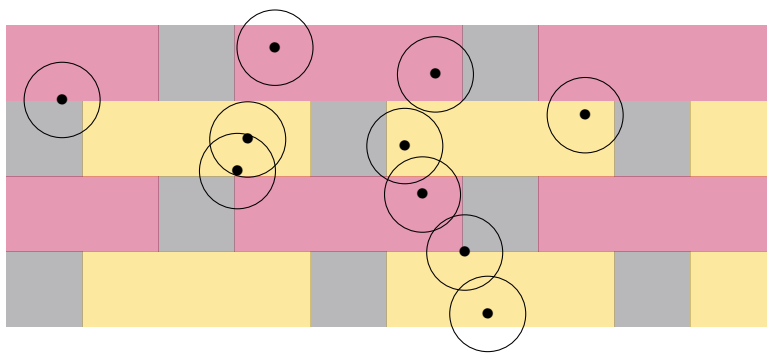

As a result, given a unit disk graph $G$ with its representation, we can partition $V(G)$ into $V_0, V_1, V_2$ based on the centers' locations in $R_0$, $R_1$, or $R_2$. Consider a unit disk graph $H$ with centers of all disks lying in a rectangle of size $1 \times \ell$. Then, $\alpha(H) \leqslant f(\ell)$ for some linear function $f$. Partition the rectangle into $4\ell$ squares of size $1/2 \times 1/2$. The graph induced by disks whose center are in such a square is a clique, hence $\alpha(H) \leqslant 4\ell$. Thus, we propose the following algorithm for input graph $G$:

\begin{enumerate}
\item Test if $G$ is a disjoint union of cliques in time $O(n+m)$.
\item Construct the partition of $V(G)$ into $V_0, V_1, V_3$ and run the FPT algorithm for testing \textsc{$2$-subcoloring} on each monochromatic connected component in time $O(n \cdot m)$. If each of these components is $2$-subcolorable, return the associated $6$-subcoloring.
\item If the last step does not yield a $6$-subcoloring, then $G$ is not $2$-subcolorable, and thus return a $7$-subcoloring in time $O(n+m)$.
\end{enumerate}
This algorithm is indeed a $3$-approximation.
\end{proof}

The next section addresses the complexity of improving the previous algorithm. Specifically, the NP-hardness of the \textsc{2-Subcoloring} problem on unit disk graphs implies that, unless $\text{P} = \text{NP}$, no polynomial-time algorithm can compute the exact subchromatic number of a unit disk graph. Moreover, achieving an approximation ratio better than $\frac{3}{2}$ is infeasible under the same assumption. One potential approach to enhance the approximation ratio from $3$ to $2$ involves the following idea: divide the plane into strips of height $1$, solve the subcoloring problem within each strip in polynomial time, and return the maximum number of colors used. It is straightforward to show that this approach results in a $2$-approximation algorithm. However, it depends critically on the existence of a polynomial-time algorithm to solve the \textsc{$k$-Subcoloring} problem on unit disk graphs of height $1$. We establish that obtaining such an algorithm would be highly challenging, as \textsc{2-Subcoloring} is NP-complete even on unit disk graphs with bounded height.

\section{Hardness results}\label{sec:hardness}
As mentioned in the introduction, \textsc{$2$-Subcoloring} is already NP-complete in various subclasses of planar graphs, such as triangle-free planar graphs. By a celebrated Theorem of Koebe, it implies that \textsc{$2$-Subcoloring} remains NP-complete in triangle-free disk graphs. 
This result can be extended to \textsc{$k$-Subcoloring} in disk graphs for every $k \geqslant 3$ by inductively doing the following reduction: take two disjoint copies of a disk graph and add a new disk adjacent to both copies. Observe then that both the subchromatic number and the maximum clique number increase by exactly one.

\subsection{\textsc{$2$-Subcoloring} of triangle-free unit disk graphs}
Unit disk graphs form a very restricted case of disk graphs. There is an infinite number of forbidden induced subgraphs (but this complete list is still unknown)~\cite{AtZa18}.
The goal of this section is to prove the first case of Theorem~\ref{thm:hardness}. Notice that the constraints are the best possible, as any triangle-free disk graph has a proper $3$-coloring. 

\begin{theorem}\label{thm:hardnessTriangleFree}
\textsc{$2$-Subcoloring} is NP-complete on triangle-free unit disk graphs.
\end{theorem}

\begin{proof}
We reduce from a variant of \textsc{Sat} called \textsc{Cubic Planar Positive 1-in-3 Sat} which is also NP-complete \cite{moore2001hard}. The input of this problem is a 3-CNF formula $\varphi$, where each clause is composed of three positive literal and where the incidence graph between variables and clauses is cubic (i.e. each vertex has degree $3$) and planar. \newline

In a $2$-subcoloring of a triangle-free graph, a color class is a disjoint union of vertices and edges. Hence, any vertex has at most one neighbor with the same color. We start by simple claims on $2$-subcoloring of cycles $C_4$ and $C_5$. 

\begin{claim}\label{claim:C5}
In any $2$-subcoloring of $C_5$, there are exactly two adjacent vertices with the same color. 
\end{claim}

\begin{proof}
Let $v_1,...,v_5$ be the $5$ vertices of a $C_5$ in cyclic order, and let $\mu : V(G) \rightarrow \{0,1\}$ be a $2$-subcoloring of $C_5$. Since $C_5$ has no proper $2$-coloring, one of the edges has to be monochromatic, say $v_1v_2$. Thus, \emph{w.l.o.g.}, we assume $\mu(v_1)=\mu(v_2)=0$. Unless creating a monochromatic $P_3$, we obtain $\mu(v_3)=\mu(v_5)=1$. Finally, $\mu(v_4)=0$ since otherwise it would create a monochromatic $P_3$.
\end{proof}

\begin{claim}\label{claim:C4}
In any $2$-subcoloring of $C_4$, if two adjacent vertices have the same color, the two other vertices have the opposite color.
\end{claim}

\begin{proof}
Let $v_1,v_2,v_3,v_4$ be the four vertices of a $C_4$ in cyclic order and suppose that $\mu(v_1)=\mu(v_2)=0$ for some $2$-subcoloring $\mu$ of $C_4$. Notice that if $\mu^(v_3)=0$, then there is a monochromatic $P_3$. Thus $\mu(v_3)=1$ and by symmetry, $\mu(v_4)=1$ too.
\end{proof}

Let $L_k$ be the ladder graph of length $k$, obtained by the cartesian product of a path of length $k$ and a path of length $2$. More formally, the ladder graph $L_k$ has vertex set $\{a_i,b_i\}_{1\leqslant i \leqslant k}$ and edges $\{a_ib_i\}_{1\leqslant i \leqslant k}\cup \{a_ia_{i+1},b_ib_{i+1}\}_{1\leqslant i <k}$. Each pair $\{a_i, b_i\}$ is called a \textit{rung} of the ladder. Using Claim~\ref{claim:C4}, we immediately obtain the property that in any $2$-subcoloring $\mu$ of $L_k$, if the two vertices of the first rung have the same color, all the rungs are monochromatic, and the two colors alternate. More formally, if $\mu(a_1)=\mu(b_1)$, then $\mu(a_i) = \mu(b_i)$ for all $1\leqslant i \leqslant k$ and $\mu(a_i) \neq \mu(a_{i+1})$ for all $1\leqslant i <k$. This graph will be used to propagate monochromatic edges.\newline

For any $k\geqslant 25$, let $F_k$ be the graph obtained by creating a ladder graph with $k$ rungs $(a_1,b_1), ...,(a_k,b_k)$, and adding an edge between $b_k$ and $b_{k-24}$. We call $a_1$ and $b_1$ the ports of the gadget $F_k$, which is depicted together with its unit disk representation in Figure~\ref{fig:Fk}. In the following lemma, we show that the ports of any graph $F_k$ cannot share the same color in any $2$-subcoloring. That is why we call it the \emph{forbidding gadget}.

\begin{claim}\label{claim:Fk}
For any forbidding gadget $F_k$ ($k\geqslant 25$) with ports $a$ and $b$ and any $2$-subcoloring $\mu$ of $F_k$, $\mu(a)\neq\mu(b)$.
\end{claim}

\begin{proof}
Let $(a_1,b_1),...,(a_k,b_k)$ be the vertices at each rung of $F_k$, where $(a,b)=(a_1,b_1)$. By contradiction, assume that $\mu(a)=\mu(b)=1$. By Claim~\ref{claim:C4}, we obtain that $\mu(a_i)=\mu(b_i)= i\mod 2$. Thus $\mu(b_k) = \mu(b_{k-24})$, and there is a monochromatic $P_3$ formed by edges $a_kb_k$ and $b_kb_{k-24}$.
\end{proof}

\begin{figure}[h]
\centering
\begin{tikzpicture}[scale=0.45]
\tikzstyle{red}=[circle,draw, fill =purple, scale=0.5]
\tikzstyle{blue}=[circle,draw, fill =lipicsYellow, scale=0.5]

\node[] () at (-0.5,0) {$a$};
\node[] () at (-0.5,-1) {$b$};

\foreach \i in {0,...,26}{
	\draw (\i/2,0) -- (\i/2,-1);
}

\draw (0,0)-- (26/2,0);
\draw (0,-1)-- (26/2,-1);

\draw (26/2,-1) to[out=-130, in = -40] (1,-1);

\foreach \i in {0,2,...,24}{
	\node[red] () at (\i/2,0) {};
	\node[red] () at (\i/2,-1) {} ;
	\node[blue] () at (\i/2+0.5,0) {};
	\node[blue] () at (\i/2+0.5,-1) {} ;
}
\node[red] () at (26/2,0) {};
\node[red] () at (26/2,-1) {} ;

\end{tikzpicture}
\quad
\begin{tikzpicture}[scale=0.5]
\foreach \i in {0,2,...,9}{
	\draw[fill=purple, fill opacity = 0.5] (\i,\i*0.2) circle (0.55);
	\draw[fill=purple, fill opacity = 0.5] (\i, \i*0.2-1) circle (0.55);
}
\foreach \i in {1,3,...,9}{
	\draw[fill=lipicsYellow, fill opacity = 0.5] (\i,\i*0.2) circle (0.55);
	\draw[fill=lipicsYellow, fill opacity = 0.5] (\i, \i*0.2-1) circle (0.55);
}

\draw[fill=purple, fill opacity = 0.5] (10, 1.5) circle (0.55);
\draw[fill=lipicsYellow, fill opacity = 0.5] (10.9, 1) circle (0.55);
\draw[fill=purple, fill opacity = 0.5] (11.5, 0.2) circle (0.55);
\draw[fill=lipicsYellow, fill opacity = 0.5] (11.9, -0.7) circle (0.55);
\draw[fill=purple, fill opacity = 0.5] (11.9, -1.7) circle (0.55);
\draw[fill=lipicsYellow, fill opacity = 0.5] (11.5, -2.6) circle (0.55);
\draw[fill=purple, fill opacity = 0.5] (10.9, -3.4) circle (0.55);
\draw[fill=lipicsYellow, fill opacity = 0.5] (10, -3.9) circle (0.55);

\draw[fill=purple, fill opacity = 0.5] (9.8, 0.6) circle (0.55);
\draw[fill=lipicsYellow, fill opacity = 0.5] (10.3, 0.3) circle (0.55);
\draw[fill=purple, fill opacity = 0.5] (10.8, -0.3) circle (0.55);
\draw[fill=lipicsYellow, fill opacity = 0.5] (10.9, -0.8) circle (0.55);
\draw[fill=purple, fill opacity = 0.5] (10.9, -1.6) circle (0.55);
\draw[fill=lipicsYellow, fill opacity = 0.5] (10.8, -2.1) circle (0.55);
\draw[fill=purple, fill opacity = 0.5] (10.3, -2.7) circle (0.55);
\draw[fill=lipicsYellow, fill opacity = 0.5] (9.8, -3) circle (0.55);

\foreach \i in {2,4,...,10}{
	\draw[fill=purple, fill opacity = 0.5] (\i-\i*0.12+2*0.12,-1-\i*0.2) circle (0.55);
	\draw[fill=purple, fill opacity = 0.5] (\i-\i*0.12+2*0.12,-1-\i*0.2-1) circle (0.55);
}

\foreach \i in {3,5,...,9}{
	\draw[fill=lipicsYellow, fill opacity = 0.5] (\i-\i*0.12+2*0.12,-1-\i*0.2) circle (0.55);
	\draw[fill=lipicsYellow, fill opacity = 0.5] (\i-\i*0.12+2*0.12,-1-\i*0.2-1) circle (0.55);
}

\end{tikzpicture}
\caption{Forbidding gadget $F_{27}$. The vertices $a$ and $b$, called the ports of the gadget, have to be of different colors in any $2$-subcoloring of $F_{27}$. Notice that the gadget can be made of arbitrary length $k\geqslant 25$.}\label{fig:Fk}
\end{figure}
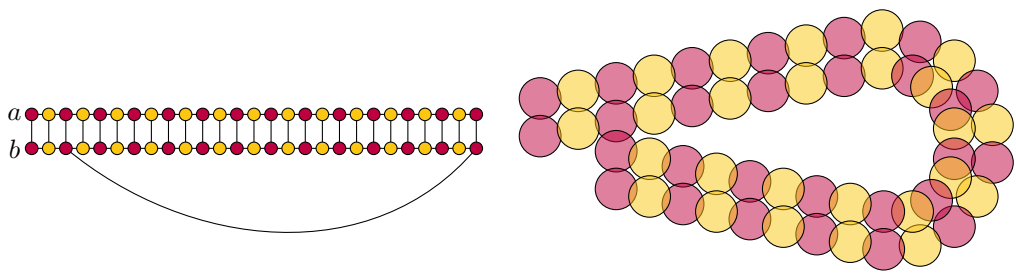

We introduce a \emph{clause gadget} $C$, as illustrated in Figure \ref{fig:C}. It is composed of a $C_5$ with vertex set $a_1,...,a_5$ in cyclic order. Then, a forbidding gadget $F_{27}$ is attached to the $C_5$ by identifying the ports of $F_{27}$ with vertices $a_4$ and $a_5$, and the same is done with vertices $a_5$ and $a_1$ this time. Using Claim~\ref{claim:C5} and Claim~\ref{claim:Fk}, we deduce that in any $2$-subcoloring of $C$, exactly one of the three edges $a_1a_2$, $a_2a_3$ and $a_3a_4$ is monochromatic. Indeed, the edges $a_4a_5$ and $a_5a_1$ cannot be monochromatic due to the presence of the two forbidding gadgets, and at least one edge of the $C_5$ must be monochromatic.\newline

\begin{figure}[!h]
\centering
\begin{tikzpicture}[scale=0.8]

\node[draw, circle, fill = purple] (a1) at (0,0) {};
\node[draw, circle, fill = purple] (a2) at (1,0) {};
\node[draw, circle, fill = lipicsYellow] (a3) at (1.2,0.8) {};
\node[draw, circle, fill = purple] (a4) at (0.5,1.3) {};
\node[draw, circle, fill = lipicsYellow] (a5) at (-0.2,0.8) {};

\draw (a1)--(a2)--(a3)--(a4)--(a5)--(a1);

\node[draw, circle] (F1) at (1.8, 2) {$F$};
\node[draw, circle] (F2) at (-0.8, 2) {$F$};
\draw[dashed] (a3) --(F1) --(a4) --(F2) --(a5);

\node[draw = none] () at (0,-0.45) {$a_2$}; 
\node[draw = none] () at (1,-0.45) {$a_3$}; 
\node[draw = none] () at (1.65,0.8) {$a_4$}; 
\node[draw = none] () at (0.5,1.75) {$a_5$}; 
\node[draw = none] () at (-0.65,0.8) {$a_1$}; 

\end{tikzpicture}
\caption{The clause gadget $C$, composed of a $C_5$ and two forbidding gadgets which ensure that both edges $a_3a_4$ and $a_4a_5$ are not monochromatic.}\label{fig:C}
\end{figure}
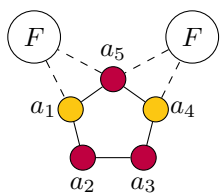

Finally, we introduce a \emph{vertex gadget} $H$. It consists of three ladders $L_1$, $L_2$, and $L_3$, of length $k_1,k_2$ and $k_3$ respectively ($k_1,k_2,k_3\geqslant 5$). For $i \in \{1, 2, 3\}$, we denote the vertices of $L_i$ as $\{a^i_j, b^i_j \mid 1 \leqslant j \leqslant k_i\}$. Additionally, we add the edges $a^1_{k_1-1}b^3_1$, $b^1_{k_1-3}b^3_1$, $a^2_{k_2-1}b^1_1$, $b^2_{k_2-3}b^1_1$, $a^3_{k_3-1}b^2_1$, and $b^3_{k_3-3}b^2_1$ to complete the construction of $H$. For $i \in \{1, 2, 3\}$, let $e_i$ denote the edge $a^i_{k_i}b^i_{k_i}$. These three edges are referred to as the \emph{ports} of $H$. A unit disk representation of $H$ of constant size is given in Figure~\ref{fig:VariableGadget}.

\begin{claim} In any $2$-subcoloring of $H$, if one of $\{e_1, e_2, e_3\}$ is monochromatic, then all of them are. \end{claim}

\begin{proof} Let $\mu$ be a $2$-subcoloring of $H$, and assume that $e_1$ is monochromatic, i.e., $\mu(a_{k_1}^1) = \mu(b_{k_1}^1)$. By Claim~\ref{claim:C4}, it follows that $\mu(a^1_{k_1-1}) = \mu(b^1_{k_1-1})$ and $\mu(a^1_{k_1-3}) = \mu(b^1_{k_1-3})$. Consequently, $\mu(a^2_1) = \mu(b^2_1)$, and by another application of Claim~\ref{claim:C4}, we deduce that $\mu(a_{k_2}^2) = \mu(b_{k_2}^2)$. Similarly, we conclude that $\mu(a_{k_3}^3) = \mu(b_{k_3}^3)$. Therefore, all the edges $e_1$, $e_2$, and $e_3$ are monochromatic. The proof is symmetric if $e_2$ or $e_3$ is monochromatic.\end{proof}

Given a $2$-subcoloring of $H$, we say that $H$ is \emph{fixed} if all its ports are monochromatic. Otherwise, $H$ is \emph{unfixed}. Note that, regardless of the colors assigned to $e_1$, $e_2$, and $e_3$, as long as none of them is monochromatic, it is always possible to extend this assignment to a $2$-subcoloring of $H$.\newline

\begin{figure}[!h]
\centering
\begin{tikzpicture}[scale=0.6, rotate = 90]

\foreach \j in {1,3,5,7,9}{
	\draw[fill = lipicsYellow, fill opacity = 0.5] (-0.2, -\j/1.2) circle (0.55);
	\draw[fill = lipicsYellow, fill opacity = 0.5] (0.8, -\j/1.2) circle (0.55);
}
\foreach \j in {2,4,6,8,10}{
	\draw[fill = purple, fill opacity = 0.5] (0, -\j/1.2) circle (0.55);
	\draw[fill = purple, fill opacity = 0.5] (1, -\j/1.2) circle (0.55);
}

\foreach \j in {-1,1,3,5,7}{
	\draw[fill = lipicsYellow, fill opacity = 0.5] (\j/1.2+0.5, -12.4/1.2) circle (0.55);
	\draw[fill = lipicsYellow, fill opacity = 0.5] (\j/1.2+0.5,-12.4/1.2+1) circle (0.55);
}
\foreach \j in {-2,0,2,4,6}{
	\draw[fill = purple, fill opacity = 0.5] (\j/1.2+0.5, -12.4/1.2-0.2) circle (0.55);
	\draw[fill = purple, fill opacity = 0.5] (\j/1.2+0.5, -12.4/1.2+0.8) circle (0.55);
}

\draw[fill = lipicsYellow, fill opacity = 0.5] (1.9, -2 ) circle (0.55);
\draw[fill = lipicsYellow, fill opacity = 0.5] (1.9, -3 ) circle (0.55);
\draw[fill = purple, fill opacity = 0.5] (2.9, -2 ) circle (0.55);
\draw[fill = purple, fill opacity = 0.5] (2.9, -3 ) circle (0.55);

\foreach \x/\y in {3.9/-2.3}{
\draw[fill = lipicsYellow, fill opacity = 0.5] (\x, \y ) circle (0.55);
\draw[fill = lipicsYellow, fill opacity = 0.5] (\x, \y-1 ) circle (0.55);
\draw[fill = purple, fill opacity = 0.5] (\x+0.9, \y-0.3 ) circle (0.55);
\draw[fill = purple, fill opacity = 0.5] (\x+0.7, \y-1.3 ) circle (0.55);
}
\draw[fill = lipicsYellow, fill opacity = 0.5] (5.7, -3 ) circle (0.55);
\draw[fill = lipicsYellow, fill opacity = 0.5] (5.3, -3.9 ) circle (0.55);

\draw[fill = purple, fill opacity = 0.5] (6.6, -3.5 ) circle (0.55);
\draw[fill = purple, fill opacity = 0.5] (6, -4.3 ) circle (0.55);

\draw[fill = lipicsYellow, fill opacity = 0.5] (7.4, -4.1 ) circle (0.55);
\draw[fill = lipicsYellow, fill opacity = 0.5] (6.6, -4.7 ) circle (0.55);

\draw[fill = purple, fill opacity = 0.5] (8.2, -4.8 ) circle (0.55);
\draw[fill = purple, fill opacity = 0.5] (7.3, -5.3 ) circle (0.55);

\draw[fill = lipicsYellow, fill opacity = 0.5] (8.4, -5.8 ) circle (0.55);
\draw[fill = lipicsYellow, fill opacity = 0.5] (7.4, -5.8 ) circle (0.55);

\draw[fill = purple, fill opacity = 0.5] (8.2, -6.8 ) circle (0.55);
\draw[fill = purple, fill opacity = 0.5] (7.2, -6.8 ) circle (0.55);

\draw[fill = lipicsYellow, fill opacity = 0.5] (8.4, -7.8 ) circle (0.55);
\draw[fill = lipicsYellow, fill opacity = 0.5] (7.4, -7.8 ) circle (0.55);

\draw[fill = purple, fill opacity = 0.5] (8.2, -8.8 ) circle (0.55);
\draw[fill = purple, fill opacity = 0.5] (7.2, -8.8 ) circle (0.55);

\draw[fill = lipicsYellow, fill opacity = 0.5] (8.4, -9.8 ) circle (0.55);
\draw[fill = lipicsYellow, fill opacity = 0.5] (7.4, -9.8 ) circle (0.55);

\draw[fill = purple, fill opacity = 0.5] (8.2, -10.8 ) circle (0.55);
\draw[fill = purple, fill opacity = 0.5] (7.2, -10.8 ) circle (0.55);

\draw[fill = lipicsYellow, fill opacity = 0.5] (8.4, -11.8 ) circle (0.55);
\draw[fill = lipicsYellow, fill opacity = 0.5] (7.4, -11.8 ) circle (0.55);

\node[draw = none] () at (0.4,0) {$e_1$};
\node[draw = none] () at (7.9,-12.8) {$e_3$};
\node[draw = none] () at (-2,-10) {$e_2$};

\end{tikzpicture}
\caption{Unit disk representation of the variable gadget $H$, with $k_1=10$, $k_2=10$ and $k_3=15$.}\label{fig:VariableGadget}
\end{figure}
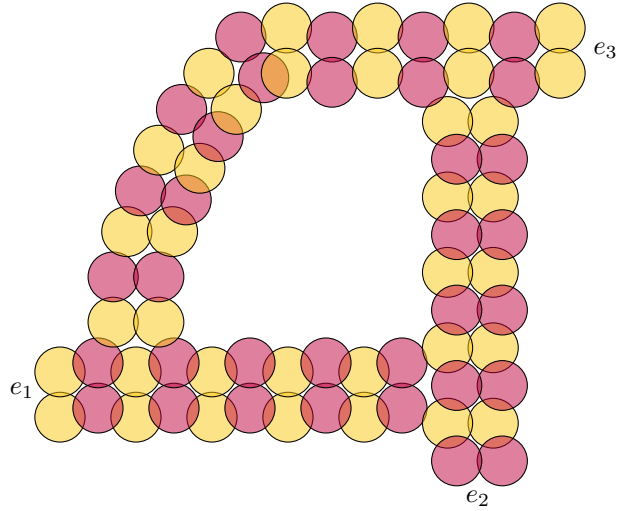

We reduce from \textsc{Cubic Planar Positive 1-in-3 SAT}. Let $\phi$ be a formula with $m$ positive clauses such that the adjacency graph $G$ is planar, and each variable $x_i$ appears in exactly $3$ clauses.

We construct a new graph $G'$ as follows. We replace each variable $x_i$ by a variable gadget $H_i$ with port edges $a_\ell^i b_\ell^i$ for $\ell \in \{1, 2,3\}$. Next, we replace each clause $c_j$ by a clause gadget $C_j$ that has three port edges $u_k^j v_k^j$ for $k = 1, 2, 3$, such that exactly one of them is monochromatic in any 2-subcoloring of $H_j$. If $x_i$ appears in clause $c_j$, we connect exactly one of the edges $a_\ell^i b_\ell^i$ to one of the $u_k^j v_k^j$ using a ladder graph. For the moment, we do not fix the size of each edge gadget which will depend on the geometry of the planar graph $G$. The following equivalence proof holds no matter their sizes.

If $G'$ is 2-subcolorable, we define the truth assignment as follows: if $H_i$ is fixed, then $x_i$ is set to true; otherwise, $x_i$ is set to false. We show that exactly one variable is set to true per clause. Let $c_j$ be a clause where exactly one of the three edges $u_k^j v_k^j$ is monochromatic. Since this edge is connected to a variable gadget $H_i$, that variable gadget is fixed, and thus $x_i$ is set to true by the truth assignment. Conversely, the other two edges are not monochromatic, so the corresponding variable gadgets are not fixed, meaning that both variables are set to false by the truth assignment. Therefore, each clause has exactly one true variable.

Conversely, if there is a truth assignment such that each clause has exactly one true variable, then for any variable $x_i$ set to true, we fix the corresponding variable gadget. In this way, exactly one edge per clause gadget is monochromatic. We then complete the 2-subcoloring of each clause gadget arbitrarily. It remains to color the variable gadgets corresponding to variables set to false. Since none of the port edges are monochromatic, we can color these gadgets arbitrarily.

Now, we show that the size of each edge gadget may be fixed such that $G'$ is a unit disk graph and has a size polynomial in the size of $G$. There exists a way to represent subdivided planar graphs with degree at most 4 as unit disk graphs, based on a result from Valiant \cite{valiant}.

\begin{theorem}[\cite{valiant}]\label{TheoremValiant}
A planar graph $G$ with maximum degree $4$ can be embedded in the plane inside a $\mathcal{O}(|V (G)|)$-sized area in such a way that any vertex is at integer coordinates and each edge is made up of vertical and horizontal line segments.
\end{theorem}
The ladder graph has high flexibility in terms of angle and size. In particular, it is possible to form an angle of $\pi/2$ using a construction similar to the forbidding gadget in Figure~\ref{fig:Fk}. Note that, in the embedding given by Theorem \ref{TheoremValiant}, the size of each edge is polynomial in $|V(G)|$. By scaling this embedding by a large enough constant, the resulting graph is indeed a unit disk graph of polynomial size. Finally, observe that the graph obtained has no triangle.
This concludes the proof.
\end{proof}

Notice that this result cannot be generalized to \textsc{$k$-Subcoloring}, as we will show in the next section that the subchromatic number of any unit disk graph is at most $7$.
A direct consequence of Theorem~\ref{thm:hardnessTriangleFree} is that \textsc{$k$-Subcoloring} cannot be approximated in unit disk graphs within a ratio of $3/2 - \varepsilon$ for any $\varepsilon > 0$ in polynomial-time, unless P=NP.

\subsection{\textsc{2-Subcoloring} of unit disk graph with bounded height}

A potential strategy for obtaining a 2-approximation algorithm on unit disk graphs involves decomposing the graph into strips of height at least 1 and then solving the \textsc{$k$-Subcoloring} problem exactly within each strip. However, we demonstrate that \textsc{$2$-Subcoloring} remains NP-complete even on unit disk graphs with bounded height. This result highly depends on constructing graphs with large cliques in the reduction. Indeed, if both the clique number and the height of a unit disk graph are bounded, it can be shown that its pathwidth is also bounded, making the problem solvable in polynomial time. The following lemma shows how a $2$-subcoloring can propagate through a sequence of cliques connected by a matching, serving as the central tool in our reduction.

\begin{lemma}\label{lemma:PropagationClique}
Let $n\geqslant 3$, and let $G = (A\cup B,E)$ be the graph such that both $G[A]$ and $G[B]$ are a clique on $n$ vertices $a_1,...,a_n$ and $b_1,...,b_n$ respectively, and such that $a_ib_j \in E$ if and only if $i=j$. Then, in any $2$-subcoloring $\mu:V(G) \rightarrow \{0,1\}$, $\mu(a_i) \neq \mu(b_i)$ for all $1\leqslant  i\leqslant n$.
\end{lemma}

\begin{proof}
By contradiction, assume that $\mu(a_i)=\mu(b_i)=1$ for some $1\leqslant i \leqslant n$. Then, for any $a\in A\setminus \{a_i\}$, $\mu(a)=0$, otherwise $aa_ib_i$ would induce a $P_3$ of color $1$. The same holds for all $b\in B\setminus \{b_i\}$. Since $n\geqslant 3$, take $1\leqslant i', i'' \leqslant n $ such that $i\neq i'$, $i \neq i''$ and $i' \neq i''$, and notice that $a_{i'}b_{i'}b_{i''}$ is a monochromatic $P_3$ of color $0$.
\end{proof}

\begin{theorem}
\textsc{2-Subcoloring} is NP-complete on unit disk graphs of bounded height.
\end{theorem}

\begin{proof}
We reduce from \textsc{Positive NAE 3-SAT}. Let $ \phi $ be a formula with $ n $ variables $ x_1, \ldots, x_n $ and $ m $ clauses $ c_1, \ldots, c_m $, where each clause contains exactly three positive literals. \emph{W.l.o.g.}, we assume that the first four variables $ x_1, \ldots, x_4 $ do not appear in any clause by introducing dummy variables if necessary. For each $ 1 \leqslant j \leqslant m $, we construct a graph $ G_j $ as follows:
\begin{itemize}
    \item For each $ 1 \leqslant i \leqslant 2n $, add two cliques $ S_{i,j} $ and $ C_{i,j} $ with vertices $ s_{i,j}^{(1)}, \ldots, s_{i,j}^{(n)} $ and $ c_{i,j}^{(1)}, \ldots, c_{i,j}^{(n)} $, respectively.
    \item For each $ 1 \leqslant i \leqslant n $, add a vertex $ \gamma_{i,j} $ adjacent to all vertices of $ C_{i,j} $ except $ c_{i,j}^{(i)} $.
    \item For each $ n+1 \leqslant i \leqslant 2n $, add a vertex $ \gamma_{i,j} $ adjacent to all vertices of $ C_{i,j} $ except $ c_{i,j}^{(2n-i+1)} $.
    \item For each $ 1 \leqslant i < 2n $, add edges between $ c_{i,j}^{(i')} $ and $ s_{i+1,j}^{(i')} $ for all $ 1 \leqslant i' \leqslant n $.
    \item For each $ 1 \leqslant i \leqslant n $, add edges $ s_{i,j}^{(i')} c_{i,j}^{(i')} $ for all $ i' \neq i $ and an edge between $ s_{i,j}^{(i)} $ and $ \gamma_{i,j} $.
    \item For each $ n+1 \leqslant i \leqslant 2n $, add edges $ s_{i,j}^{(i')} c_{i,j}^{(i')} $ for all $ i' \neq 2n-i+1 $ and an edge between $ s_{i,j}^{(2n-i+1)} $ and $ \gamma_{i,j} $.
\end{itemize}

In Figure~\ref{fig:BoundedHeightG1}, an illustration of the construction of $G_j$ is given.
\begin{figure}[h]
\centering
\begin{tikzpicture}[scale=0.8]

\pgfmathsetmacro{\n}{4};
\pgfmathsetmacro{\eps}{0.3};

\foreach \i in {1,...,\n}{
	\foreach \j in {1,...,\n}{
		\node[draw, circle, fill] () at (2*\i, -\j -\i) {};
		\draw (2*\i, -\j -\i) -- (2*\i+1, -\j -\i);
	}
	\draw (2*\i-\eps, -\n-\i-\eps) rectangle (2*\i+\eps, -1-\i+\eps);
	\foreach \j in {1,...,\n}{
		\node[draw, circle, fill] () at (2*\i+1,-{\j} -\i) {};
	}
	\node[draw, circle, fill] () at (2*\i+1,-\n-1 -\i) {};

	\draw (2*\i+1-\eps, -\n-1-\i-\eps) rectangle (2*\i+1+\eps, -2-\i+\eps);
	
	\node[draw = none, label = $\gamma_{\i,j}$] () at (2*\i+1,-1 -\i) {};
	\foreach \j in {2,...,\n}{
		\draw (2*\i+1,-1 -\i) to[out = -45, in = 45] (2*\i+1,-\j -\i) ;
	}
	
	\foreach \j in {2,...,5}{
		\draw (2*\i+1,-{\j} -\i) -- (2*\i+2,-{\j} -\i) ;
	}
}

\foreach \i in {5,...,8}{
	\foreach \j in {1,...,\n}{
		\node[draw, circle, fill] () at (2*\i, -\j+\i-2*\n-2) {};
		\node[draw, circle, fill] () at (2*\i+1, -\j+\i-2*\n-2) {};
		\draw (2*\i, -\j+\i-2*\n-2) -- (2*\i+1, -\j+\i-2*\n-2);

	}
\node[draw = none, label = $\gamma_{\i,j}$] () at (2*\i+1,\i-4*\n+1) {};

	\draw (2*\i-\eps, -\n+\i-2*\n-2-\eps) rectangle (2*\i+\eps, -1+\i-2*\n-2+\eps);
	\draw (2*\i+1-\eps, -\n+\i-2*\n-2+1-\eps) rectangle (2*\i+1+\eps, -1+\i-2*\n-2+1+\eps);
	\node[draw, circle, fill] () at (2*\i+1, \i-2*\n-2) {};
	
	\foreach \j in {1,...,3}{
		\draw (2*\i+1, -\n+\i-2*\n-2) to[out = 45, in = -45] (2*\i+1, -\j+\i-2*\n-2) ;
	}

}
\foreach \i in {6,...,8}{
	\foreach \j in {1,...,\n}{
		\draw (2*\i, -\j+\i-2*\n-2)  -- (2*\i-1, -\j+\i-2*\n-2) ;
	}
}
\end{tikzpicture}
\caption{Graph $G_j$ for a formula with $4$ variables.}\label{fig:BoundedHeightG1}
\end{figure}
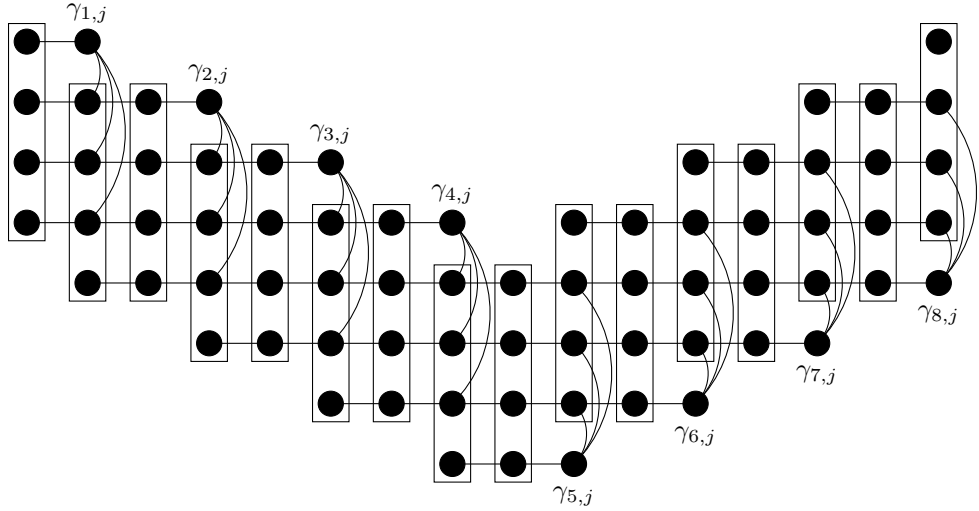

Next, let $ G $ be the graph constructed as follows:
\begin{itemize}
    \item Take the disjoint union of $ G_1, \ldots, G_m $ and add edges $ c_{2n,j}^{(i')} s_{1,j+1}^{(i')} $ for all $ 1 \leqslant j < m $ and $ 1 \leqslant i' \leqslant n $.
    \item For each $ 1 \leqslant j \leqslant m $, add a path on three vertices $ u_j^{(1)} u_j^{(2)} u_j^{(3)} $.
    \item For each $ 1 \leqslant j \leqslant m $, let $ 1 \leqslant i_1 < i_2 < i_3 \leqslant n $ be the indices corresponding to the variables in the clause $ c_j = x_{i_1} \vee x_{i_2} \vee x_{i_3} $. Add three ladders $ L_{j,1}, L_{j,2}, L_{j,3} $ (a ladder is defined in the proof of Theorem~\ref{thm:hardnessTriangleFree}). For each $\ell \in \{1, 2, 3\}$:
    \begin{itemize}
        \item Connect the two vertices of the first rung of ladder $ L_{j, \ell} $ to $ \gamma_{i_\ell, j} $.
        \item Connect one of the two vertices of the last rung of ladder $ L_{j, \ell} $ to $ u_j^{(\ell)} $.
        \item The length of each ladder is odd, and upper bounded by a polynomial function of $ n, i_1, i_2 $, and $ i_3 $, to be determined later.
    \end{itemize}
    \item Add a path on three vertices $a_1a_2a_3$ such that $a_1$ is adjacent to $ s_{1,1}^{(1)} $, $ a_2$ is adjacent to $ s_{1,1}^{(1)} $ and $s_{1,1}^{(2)}$ and $ a_3$ is adjacent to $ s_{1,1}^{(3)} $ and $s_{1,1}^{(4)}$
\end{itemize}

An illustration of the graph $G$ is given in Figure~\ref{fig:BoundedHeightG}.
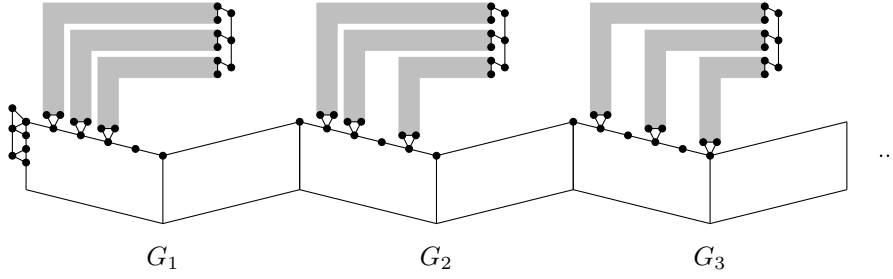
\begin{figure}
\centering
\begin{tikzpicture}[scale=0.9]

\foreach \i in {12.5,12.6,12.7}{
	\draw[fill = black] (\i, 0.5) circle (0.01);
}

\foreach \j/\x in {0.4/2.6,0.8/2.2,1.2/1.8}{
\draw[gray, line width = 8, opacity = 0.5] (\j,1-\j/4+0.2) -- (\j,\x) -- (3-0.2,\x);
\draw[fill = black] (\j-0.1, 1-\j/4+0.2) circle (0.05);
\draw[fill = black] (\j+0.1, 1-\j/4+0.2)  circle (0.05);
\draw (\j, 1-\j/4) -- (\j-0.1, 1-\j/4+0.2) -- (\j+0.1, 1-\j/4+0.2) -- (\j, 1-\j/4);

}

\foreach \j/\x in {0.4/2.6,0.8/2.2,1.6/1.8}{
\draw[gray, line width = 8, opacity = 0.5] (4+\j,1-\j/4+0.2) -- (4+\j,\x) -- (7-0.2,\x);
\draw[fill = black] (4+\j-0.1, 1-\j/4+0.2) circle (0.05);
\draw[fill = black] (4+\j+0.1, 1-\j/4+0.2)  circle (0.05);
\draw (4+\j, 1-\j/4) -- (4+\j-0.1, 1-\j/4+0.2) -- (4+\j+0.1, 1-\j/4+0.2) -- (4+\j, 1-\j/4);

}

\foreach \j/\x in {0.4/2.6,1.2/2.2,2/1.8}{
\draw[gray, line width = 8, opacity = 0.5] (8+\j,1-\j/4+0.2) -- (8+\j,\x) -- (11-0.2,\x);
\draw[fill = black] (8+\j-0.1, 1-\j/4+0.2) circle (0.05);
\draw[fill = black] (8+\j+0.1, 1-\j/4+0.2)  circle (0.05);
\draw (8+\j, 1-\j/4) -- (8+\j-0.1, 1-\j/4+0.2) -- (8+\j+0.1, 1-\j/4+0.2) -- (8+\j, 1-\j/4);

}

\foreach \i in {0,4,8}{
\draw (\i,1) -- (\i+2, 0.5) -- (\i+4,1);
\draw (\i,0) -- (\i+2, -0.5) -- (\i+4,0);
\draw (\i,0)--(\i,1);
\draw (\i+4,0)--(\i+4,1);
\draw (\i+2,-0.5) -- (\i+2,0.5);

\foreach \j in {0,0.4,0.8,...,2.2}{
	\draw[fill = black] (\i+\j, 1-\j/4) circle (0.05);
}

\foreach \j in {1.8,2.2,2.6}{
	\draw[fill = black] (\i+3, \j) circle (0.05);
	\draw[fill = black] (\i+3-0.2, \j-0.1) circle (0.05);
	\draw[fill = black] (\i+3-0.2, \j+0.1) circle (0.05);
	\draw (\i+3-0.2, \j-0.1)  -- (\i+3-0.2, \j+0.1) -- (\i+3, \j) ;
	
}
\draw (\i+3,1.8) -- (\i+3,2.6);
}

\foreach \i in {1,2,3}{
	\node[draw = none] () at (4*\i- 2, -1) {$G_\i$};
}

\foreach \j in {0.4,0.6,0.8,1}{
	\draw[fill = black] (0, \j) circle (0.05);
}
\foreach \j in {0.5,0.9,1.2}{
	\draw[fill = black] (-0.2, \j) circle (0.05);
}
\draw (-0.2,0.5) -- (-0.2,1.2);
\foreach \j in {0.4,0.6}{
	\draw (0, \j) -- (-0.2,0.5);
}
\foreach \j in {0.8 ,1}{
	\draw (0, \j) -- (-0.2,0.9);
}
\draw (0, 1) -- (-0.2,1.2);

\end{tikzpicture}
\caption{Illustration of the graph $G$, where the thick gray lines are ladders.}\label{fig:BoundedHeightG}
\end{figure}

We prove that $G$ has a $2$-subcoloring if, and only if, $\phi$ has a truth assignment such that each clause contains at least one or exactly two true variables.

Assume that $G$ has a $2$-subcoloring $\mu : V(G) \rightarrow \{0,1\}$. First, notice that $\mu(s_{1,1}^{(1)}) = \mu(s_{1,1}^{(2)})$. If it is not the case, then if we assume that $\mu(a_2)=0$, then all the vertices from $S_{1,1}\setminus \{s_{1,1}^{(1)}, s_{1,1}^{(2)}\}$ have color $1$, and notice that $a_3$ cannot have color $0$ or $1$ without creating a monochromatic $P_3$. Similarly, we have that $\mu(s_{1,1}^{(3)}) = \mu(s_{1,1}^{(4)})$. By contradiction, if $\mu(s_{1,1}^{(1)}) = \mu(s_{1,1}^{(3)})=1$, then $a_2$ and $a_3 $ have color $0$, and $a_1$ cannot have color $0$ or $1$. It follows that $\mu$ uses both colors $0$ and $1$ at least twice on $S_{1,1}$. By Lemma~\ref{lemma:PropagationClique}, for all $2 \leqslant i' \leqslant n$, we have $\mu(c_{1,1}^{(i')}) = 1 - \mu(s_{1,1}^{(i')})$ and $\mu(\gamma_{1,1}) = 1 - \mu(s_{1,1}^{(1)})$.
We prove that $\mu(c_{1,1}^{(1)} ) \neq \mu(\gamma_{1,1})$. Assume by contradiction that $\mu(c_{1,1}^{(1)}) = \mu(\gamma_{1,1}) = 1$. Then, all vertices of $C_{1,1} \setminus \{c_{1,1}^{(1)}\}$ have color $0$. However, at least one of $c_{1,1}^{(2)}, c_{1,1}^{(3)}, c_{1,1}^{(4)}$ must have color $1$, which is a contradiction.  By Lemma~\ref{lemma:PropagationClique}, it follows that for all $2 \leqslant i' \leqslant n$, we have $
\mu(s_{2,1}^{(i')}) = 1 - \mu(c_{2,1}^{(i')}) = \mu(s_{1,1}^{(i')})$ and $\mu(s_{2,1}^{(1)}) = 1 - \mu(s_{1,1}^{(1)})$.

The same reasoning allows us to propagate the colors along $G_1, \ldots, G_m$. In particular, for any $1 \leqslant j \leqslant m$ and $1 \leqslant i \leqslant n$, we have $\mu(s_{1,j}^{(i)}) = \mu(s_{1,1}^{(i)})$ and $\mu(\gamma_{i,j}) = 1 - \mu(s_{1,1}^{(i)})$.

Consider the following truth assignment: set the variable $x_i$ to true if, and only if, $\mu(s_{1,1}^{(i)}) = 1$. Now, assume by contradiction that one of the clauses, say $c_1$, contains only true variables. The case where all variables are false is symmetric. Let $c_1 = x_{i_1} \vee x_{i_2} \vee x_{i_3}$. Then, for all $\ell \in \{1, 2, 3\}$, we have $\mu(\gamma_{i_\ell}) = 0$. Since $C_{i_\ell,j}$ has at least two vertices of color $0$, the two vertices of the first rung of the ladder $L_{j,\ell}$ have color $1$. By the property of a ladder of odd length (see Claim~\ref{claim:C4} of Theorem~\ref{thm:hardnessTriangleFree}), the last rungs of $L_{j,1}, L_{j,2}$, and $L_{j,3}$ have the same color $1$. Since the three vertices $u_j^{(1)}, u_j^{(2)}, u_j^{(3)}$ are all connected to only one vertex of the last rung of $L_{j_1},L_{j_2}$ and $L_{j_3}$ respectively, they all have color $0$, making $\mu$ not a subcoloring.

Conversely, assume that $\phi$ has a truth assignment such that no clause contains only true or only false variables. For all $5 \leqslant i' \leqslant n$, set $\mu(s_{1,1}^{(i')}) = 1$ if and only if $x_{i'}$ is set to true. Then, set $\mu(s_{1,1}^{(1)}) = \mu(s_{1,1}^{(3)}) = 0$ and $\mu(s_{1,1}^{(2)}) = \mu(s_{1,1}^{(4)}) = 1$. Since both colors $0$ and $1$ are used at least twice on the first clique, there is a unique way to propagate the colors along $G_1, \ldots, G_m$ as shown in the forward direction of the equivalence.

Similarly, there is a unique way to propagate the colors along the ladders. Finally, for each $1 \leqslant j \leqslant m$, it is possible to assign colors to $u_j^{(1)}, u_j^{(2)}, u_j^{(3)}$ without creating a monochromatic path on three vertices since the last rungs of the three ladders do not have the same color.\newline

Finally, we show that $G$ is a unit disk graph. Let $\varepsilon = \frac{1}{n^2}$. We provide explicit centers for each vertex in $G_1, \ldots, G_m$. For $1 \leqslant j \leqslant m$ and $1 \leqslant i \leqslant n$, we define the centers as follows:

\begin{itemize}
    \item For each $i \leqslant i' \leqslant n$, place a disk at coordinates $(4n(j-1) + 2(i-1), -i'\varepsilon)$ corresponding to $s_{i,j}^{(i')}$, and a disk at $(4n(j-1) + 2(i-1) + 1, -i'\varepsilon)$. This disk corresponds to $c_{i,j}^{(i')}$ if $i' > i$, and to $\gamma_{i,j}$ if $i' = i$.
    
    \item For each $1 \leqslant i' < i$, place a disk at $(4n(j-1) + 2(i-1), 1 - \varepsilon/2 - i'\varepsilon)$ corresponding to $s_{i,j}^{(i')}$, and a disk at $(4n(j-1) + 2(i-1) + 1, 1 - \varepsilon/2 - i'\varepsilon)$ corresponding to $c_{i,j}^{(i')}$.
    
    \item Finally, add a disk at $(4n(j-1) + 2(i-1) + 1, 1 - \varepsilon/2 - i\varepsilon)$ corresponding to $c_{i,j}^{(i)}$.
\end{itemize}

It can be observed that this construction provides a valid disk representation for the subgraph induced by the vertices in $V(G_1), \ldots, V(G_m)$. An illustration of the partial disk representation is shown in Figure~\ref{fig:BoundedHeightDiskRepr}.

Next, we can add three disks for the vertices $a_1$, $a_2$, and $a_3$, ensuring their respective adjacencies are respected. Additionally, observe that the disks representing the vertices $\gamma_{i,j}$ are \emph{accessible}, meaning there is sufficient space to add new disks adjacent exclusively to these vertices. Thus, for any $1 \leqslant j \leqslant m$, we can attach three ladders to the three $\gamma_{i,j}$'s (where $x_i$ appears in $c_j$). Using long enough ladders, it is possible to attach them to a $P_3$ to complete the final disk representation of $G$, while keeping the height of the representation bounded by some constant that we do not exhibit (but is at least $2$).

\begin{figure}
\centering
\begin{tikzpicture}[scale=1.5]

\foreach \ip in {1,2,3,4}{
	\draw[fill = red, fill opacity =0.2] (0,-\ip/6) circle (1/2);
	\draw[fill =blue, fill opacity =0.2]  (1,-\ip/6) circle (1/2);
}
\draw[fill = red, fill opacity =0.2]  (1,-1-1/18 -1/6) circle (1/2);

\foreach \ip in {2,3,4}{
	\draw[fill = red, fill opacity =0.2] (2,-\ip/6) circle (1/2);
	\draw[fill = blue, fill opacity =0.2]  (3,-\ip/6) circle (1/2);
}
\draw[fill = blue, fill opacity =0.2]  (2,-1-1/18 -1/6) circle (1/2);

\foreach \ip in {1,2}{
	\draw[fill = red, fill opacity =0.2]  (3,-1-1/18 -\ip/6) circle (1/2);

}

\draw[fill = red, fill opacity =0.2] (4,-3/6) circle (1/2);
\draw[fill = red, fill opacity =0.2] (4,-4/6) circle (1/2);

\foreach \ip in {1,2}{
	\draw[fill = blue, fill opacity =0.2]  (4,-1-1/18 -\ip/6) circle (1/2);

}

\foreach \ip in {3,4}{
\draw[fill = blue, fill opacity =0.2] (5,-\ip/6) circle (1/2);
}
\draw[fill = red, fill opacity =0.2] (6,-4/6) circle (1/2);

\foreach \ip in {1,2,3}{
	\draw[fill = red, fill opacity =0.2]  (5,-1-1/18 -\ip/6) circle (1/2);
	\draw[fill = blue, fill opacity =0.2]  (6,-1-1/18 -\ip/6) circle (1/2);
	\draw[fill = red, fill opacity =0.2]  (7,-1-1/18 -\ip/6) circle (1/2);
}
\draw[fill = red, fill opacity =0.2]  (7,-1-1/18 -4/6) circle (1/2);
\draw[fill = blue, fill opacity =0.2] (7,-4/6) circle (1/2);

\end{tikzpicture}
\caption{Partial disk representation of $G_1$ for the first three variables.}\label{fig:BoundedHeightDiskRepr}
\end{figure}
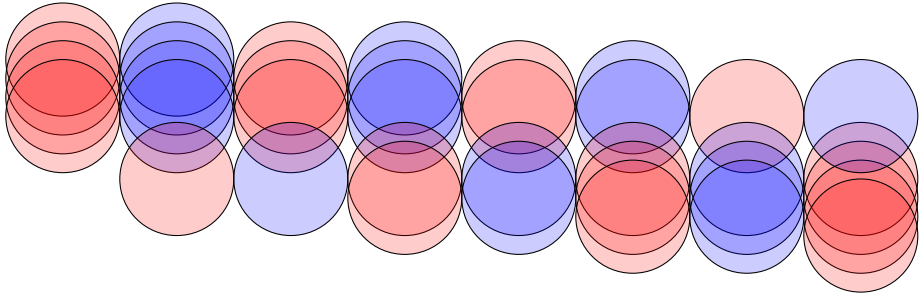

\end{proof}

It remains an open question whether \textsc{2-Subcoloring} is NP-complete for unit disk graphs of height at most $1$. Resolving this question would definitively determine whether a $2$-approximation algorithm can be ruled out using a strip decomposition of the input unit disk graph.

\subsection{\textsc{$k$-Subcoloring} of co-comparability graphs}

The previous proof highlights the main idea for showing the hardness of \textsc{$k$-Subcoloring} on co-comparability graphs: using sequences of cliques to propagate coloring constraints. While the general approach is similar, the technical details differ. We reduce from the \textsc{No Rainbow Hypergraph 3-Coloring} problem, where the goal is to find a $3$-coloring of a $3$-uniform hypergraph such that no edge contains all three colors, and specific vertices must be assigned fixed colors.\newline

\fbox{
\begin{minipage}{0.9\linewidth}
\textsc{No Rainbow Hypergraph 3-Coloring :}

    \textbf{Input:} A $3$-uniform hypergraph $ H = (V, \mathcal{E}) $ with vertex set $ V $ and edge set $\mathcal{E}$ and $3$ sets $A_1,A_2,A_3 \subseteq V$.

    \textbf{Question:}  
    Is there a $3$-coloring $\mu : V \rightarrow \{1,2,3\}$ such that no hyperedge $ e \in \mathcal{E} $ is "rainbow" (i.e., no hyperedge contains all 3 colors) and such that for any $k \in \{1,2,3\}$ and $v\in A_k$, $\mu(v) = k$ ?
\end{minipage}}

\begin{theorem}[\cite{bulatov2006dichotomy}]
    \textsc{No Rainbow Hypergraph 3-Coloring} is NP-complete.
\end{theorem}
In \cite{zhuk2021no}, it was proven that the same problem without fixing some vertices with some color but just requiring that the $3$-coloring is surjective, is also NP-complete. We start with a straighforward lemma.

\begin{lemma}\label{lemma:K444}
    In any $3$-subcoloring of $K_{4,4,4}$, two vertices of two different parts have different colors, or equivalently, the three parts are all monochromatic with three different colors.
\end{lemma}

\begin{proof}
Let $A$, $B$, and $C$ be the three distinct parts of a $K_{4,4,4}$ graph, and let $\mu$ be a 3-subcoloring of this graph. Suppose, for contradiction, that there exist two vertices in different parts with the same color. \emph{W.l.o.g.}, assume $a \in A$ and $b \in B$ are both colored 3. Then, no other vertex $a' \in A \setminus \{a\}$ can have color 3; otherwise, the path $a b a'$ would form a monochromatic path of three vertices. Similarly, no other vertex $b' \in B \setminus \{b\}$ can have color 3.

Now, consider part $C$. If two different vertices $c, c' \in C$ are also colored 3, then the path $c a c'$ would also form a monochromatic path of three vertices. Consequently, after removing all vertices colored 3 from the graph, each part still contains at least three vertices. This configuration yields a 2-subcoloring of $K_{3,3,3}$, which cannot exist.
\end{proof}

\begin{theorem}
\textsc{3-Subcoloring} is NP-complete on co-comparability graphs.
\end{theorem}

\begin{proof}
We reduce from \textsc{No Rainbow Hypergraph 3-Coloring}. Let $H$ be a 3-uniform hypergraph with $n$ vertices $v_1, \dots, v_n$ and $m$ hyperedges. Let $A_1, A_2, A_3 \subseteq V(H)$. We construct a co-comparability graph $G$ as follows:

\begin{enumerate}
    \item For each $j = 1, \dots, m$, create two copies of $K_{n+6}$, denoted $C_j = \{c_{j,1}, c_{j,2}, \dots, c_{j,n+6}\}$ and $L_j = \{l_{j,1}, l_{j,2}, \dots, l_{j,n+6}\}$.

    \item For each $j = 1, \dots, m-1$, add edges to form a matching between $C_j$ and $L_j$, and similarly between $L_j$ and $C_{j+1}$. Specifically, add edges $\{(c_{j,i}, l_{j,i}) \mid 1 \leqslant i \leqslant n+6\}$ and $\{(l_{j,i}, c_{j+1,i}) \mid 1 \leqslant i \leqslant n+6\}$.

    \item For each $j = 1, \dots, m$, add a set of $n$ vertices $\{l_{j,i}' \mid 1 \leqslant i \leqslant n+6\}$ to $L_j$, where each $l_{j,i}'$ is a false twin of $l_{j,i} \in L_j$. More precisely, each $l_{j,i}'$ is adjacent to:
    \begin{itemize}
        \item all vertices in $L_j$ except $l_{j,i}$,
        \item the vertices $c_{j,i} \in C_j$ and, if $j < m$, $c_{j+1,i} \in C_{j+1}$.
    \end{itemize}
	Now, $L_j$ denote the set of vertices $\bigcup_{1\leqslant i\leqslant n} \{l_{j,i},l_{j,i}'\}$.
    \item For each hyperedge $e_j \in \mathcal{E}(H)$, introduce a new vertex $s_j$ (the "selector" vertex).
    \begin{itemize}
        \item Connect $s_j$ to all vertices in $L_j$, i.e., add edges $\{(s_j, u) \mid u \in L_j\}$.
        \item Additionally, connect $s_j$ to the vertices $\{c_{j,i} \mid v_i \in e_j\}$ in $C_j$ that correspond to the vertices in the hyperedge $e_j$.
    \end{itemize}

    \item Add a complete multipartite graph $K_{4,4,4}$ with parts $B_1$, $B_2$, and $B_3$. For each $k \in \{1,2,3\}$ and $v_i \in A_k$, add edges from $c_{1,i}$ to each vertex in $B_\ell$ for $\ell \in \{1,2,3\} \setminus \{k\}$. Additionally, add edges from $\{c_{1,n+1}, c_{1,n+2}\}$ to $B_2 \cup B_3$, from $\{c_{1,n+3}, c_{1,n+4}\}$ to $B_1 \cup B_3$, and from $\{c_{1,n+5}, c_{1,n+6}\}$ to $B_1 \cup B_2$.
\end{enumerate}

Observe that the size of $G$ is polynomial in the size of $H$, and we also confirm that $G$ is a co-comparability graph.

\begin{claim}
The graph $ G $ constructed above is a co-comparability graph.
\end{claim}

\begin{proof}
   In Figure~\ref{fig:StringRepresentation}, we illustrate a string representation for a simple example of $G$. Below, we provide a formal proof of the desired property. The class of co-comparability graphs is closed under the addition of false twins. Indeed, in any string representation of such a graph, a string can be replaced by two non-crossing strings following the original one, resulting in the creation of false twins. Therefore, to establish that $G$ is a co-comparability graph, it suffices to show that the graph obtained by contracting each pair of false twins in $G$ is a co-comparability graph.

    Let $G'$ be the graph obtained by contracting each pair of false twins $\{l_{j,i}, l_{j,i}'\}$ into a single vertex $t_{j,i}$ for all $1 \leq j \leq m$ and $1 \leq i \leq n+6$, and for each $k \in \{1,2,3\}$, contracting $B_k$ into a single vertex $b_k$. Define the $2m+1$ cliques of $G'$ as $D_0, \dots, D_{2m}$, where $D_0 = \{b_1, b_2, b_3\}$, $D_{2j-1} = C_j$, and $D_{2j} = \{t_{j,i} \mid i = 1, \dots, n+6\} \cup \{s_j\}$ for $1 \leq j \leq m$. Observe that every edge of $G'$ is either within one of these cliques or between two successive cliques.
    
    Now, consider the following order:
    \begin{align*}
        &b_1 \prec b_2 \prec b_3 \prec c_{1,1} \prec \dots \prec c_{1,n+6} \prec s_1 \prec t_{1,1} \prec \dots \prec t_{1,n+6} \prec \dots \\
        \prec &c_{m,1} \prec \dots \prec c_{m,n+6} \prec s_m \prec t_{m,1} \prec \dots \prec t_{m,n+6}.
    \end{align*}
    
    We will show that this ordering respects the non-edges of $G'$. Suppose, for some $u, v, w \in V(G')$, that $u \prec v$, $v \prec w$, $uv \notin E(G')$, and $vw \notin E(G')$. Assume, for contradiction, that $uw \in E(G')$. Then, $u$ and $w$ must either belong to the same clique or to two consecutive cliques. If $u, w \in D_j$ for some $j$, then by the construction of the order, $v \in D_j$, implying $uv \in E(G')$, a contradiction. If $u \in D_j$ and $w \in D_{j+1}$ for some $j$, then, by the ordering, $v$ must be in either $D_j$ or $D_{j+1}$. In the first case, $uv \in E(G')$, and in the second, $vw \in E(G')$, both contradicting the assumption. This confirms that $G'$ respects the co-comparability ordering.
\end{proof}

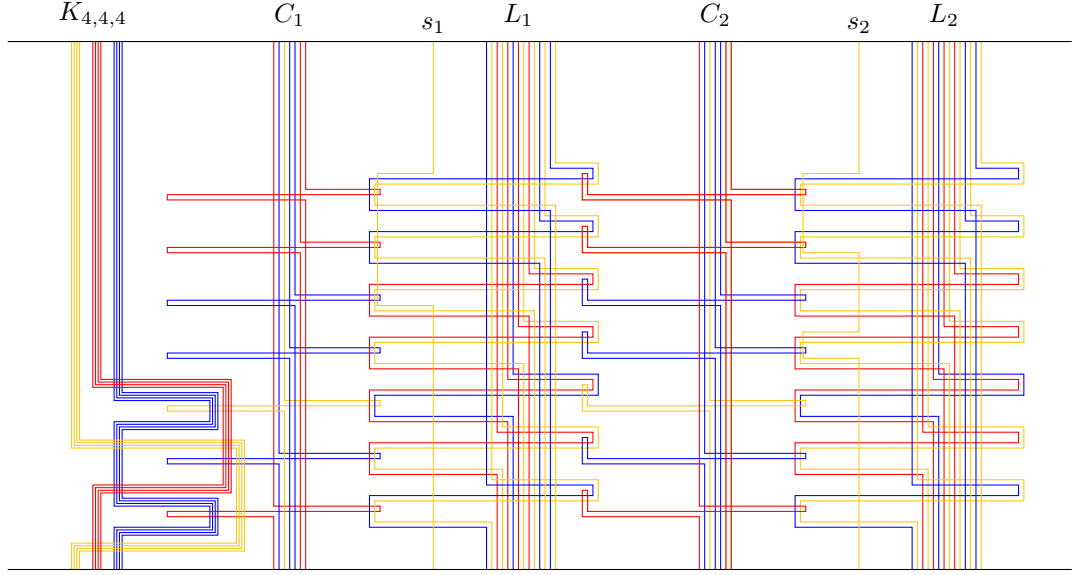
\begin{figure}[h]
    \centering
    \begin{tikzpicture}[scale=0.7]
    \draw[red] (0,0) -- (0,1) -- (-2, 1) -- (-2,1.1) -- (2,1.1) -- (2,1.2) -- (0,1.2) -- (0,10);
    \draw[blue] (0.1,0) -- (0.1,2) -- (-2, 2) -- (-2,2.1) -- (2,2.1) -- (2,2.2) -- (0.1,2.2) -- (0.1,10);
    \draw[lipicsYellow] (0.2,0) -- (0.2,3) -- (-2, 3) -- (-2,3.1) -- (2,3.1) -- (2,3.2) -- (0.2,3.2) -- (0.2,10);

    \foreach \i in {4,5}{
        \draw[blue] (\i*0.1-0.1,0) -- (\i*0.1-0.1,\i) -- (-2, \i) -- (-2,\i+0.1) -- (2,\i+0.1) -- (2,\i+0.2) -- (\i*0.1-0.1,\i+0.2) -- (\i*0.1-0.1,10);
    }
    \foreach \i in {6,7}{
        \draw[red] (\i*0.1-0.1,0) -- (\i*0.1-0.1,\i) -- (-2, \i) -- (-2,\i+0.1) -- (2,\i+0.1) -- (2,\i+0.2) -- (\i*0.1-0.1,\i+0.2) -- (\i*0.1-0.1,10);
    }

    \foreach \i in {0,1,2,3}{
        \draw[blue] (-3+\i*0.05,0) -- (-3+\i*0.05,0.8-\i*0.05) -- (-1.2+\i*0.05,0.8-\i*0.05) -- (-1.2+\i*0.05,1.2+\i*0.05) -- (-3+\i*0.05,1.2+\i*0.05) -- (-3+\i*0.05,2.8-\i*0.05) -- (-1.2+\i*0.05,2.8-\i*0.05) -- (-1.2+\i*0.05,3.2+\i*0.05) -- (-3+\i*0.05,3.2+\i*0.05) -- (-3+\i*0.05,10);

    }

    \foreach \i in {0,1,2,3}{
        \draw[lipicsYellow] (-3.8+\i*0.05,0) -- (-3.8+\i*0.05,0.5-\i*0.05) -- (-0.7+\i*0.05,0.5-\i*0.05) -- (-0.7+\i*0.05,2.3+\i*0.05) -- (-3.8+\i*0.05,2.3+\i*0.05) -- (-3.8+\i*0.05,10);

    }

    \foreach \i in {0,1,2,3}{
        \draw[red] (-3.4+\i*0.05,0) -- (-3.4+\i*0.05,1.6-\i*0.05) -- (-0.95+\i*0.05,1.6-\i*0.05) -- (-0.95+\i*0.05,3.45+\i*0.05) -- (-3.4+\i*0.05,3.45+\i*0.05) -- (-3.4+\i*0.05,10);

    }


    \foreach \j in {4,12}{
    
        \draw[red] (\j+0.4,0) -- (\j+0.4,2.8) -- (\j-2.2, 2.8) -- (\j-2.2,3.4) -- (\j+2,3.4) -- (\j+2,3.6) -- (\j+0.4,3.6) -- (\j+0.4,10);
        \draw[blue] (\j.5,0) -- (\j.5,2.9) -- (\j-2.1, 2.9) -- (\j-2.1,3.3) -- (\j+2.1,3.3) -- (\j+2.1,3.7) -- (\j.5,3.7) -- (\j.5,10);
    
        \foreach \i in {1,3,4}{
            \draw[red] (\j+\i*0.2,0) -- (\j+\i*0.2,\i+0.8) -- (\j-2.2, \i+0.8) -- (\j-2.2,\i+1.4) -- (\j+2,\i+1.4) -- (\j+2,\i+1.6) -- (\j+\i*0.2,\i+1.6) -- (\j+\i*0.2,10);
            \draw[lipicsYellow] (\j+\i*0.2+0.1,0) -- (\j+\i*0.2+0.1,\i+0.9) -- (\j-2.1, \i+0.9) -- (\j-2.1,\i+1.3) -- (\j+2.1,\i+1.3) -- (\j+2.1,\i+1.7) -- (\j+\i*0.2+0.1,\i+1.7) -- (\j+\i*0.2+0.1,10);
        }
    
        \foreach \i in {0,5,6}{
            \draw[blue] (\j+\i*0.2,0) -- (\j+\i*0.2,\i+0.8) -- (\j-2.2, \i+0.8) -- (\j-2.2,\i+1.4) -- (\j+2,\i+1.4) -- (\j+2,\i+1.6) -- (\j+\i*0.2,\i+1.6) -- (\j+\i*0.2,10);
            \draw[lipicsYellow] (\j+\i*0.2+0.1,0) -- (\j+\i*0.2+0.1,\i+0.9) -- (\j-2.1, \i+0.9) -- (\j-2.1,\i+1.3) -- (\j+2.1,\i+1.3) -- (\j+2.1,\i+1.7) -- (\j+\i*0.2+0.1,\i+1.7) -- (\j+\i*0.2+0.1,10);
        }
    }

    \foreach \i in {3,6,7}{
        \draw[lipicsYellow] (8+\i*0.1-0.1,0) -- (8+\i*0.1-0.1, \i) -- (8-2.2, \i) -- (8-2.2,\i+0.5) -- (8-2.1,\i+0.5) -- (8-2.1,\i+0.1) -- (8+2,\i+0.1) -- (8+2,\i+0.2) -- (8+\i*0.1-0.1, \i+0.2) -- (8+\i*0.1-0.1, 10);
    }

    \foreach \i in {2,4,5}{
        \draw[blue] (8+\i*0.1-0.1,0) -- (8+\i*0.1-0.1, \i) -- (8-2.2, \i) -- (8-2.2,\i+0.5) -- (8-2.1,\i+0.5) -- (8-2.1,\i+0.1) -- (8+2,\i+0.1) -- (8+2,\i+0.2) -- (8+\i*0.1-0.1, \i+0.2) -- (8+\i*0.1-0.1, 10);
    }
    \foreach \i in {1,6,7}{
        \draw[red] (8+\i*0.1-0.1,0) -- (8+\i*0.1-0.1, \i) -- (8-2.2, \i) -- (8-2.2,\i+0.5) -- (8-2.1,\i+0.5) -- (8-2.1,\i+0.1) -- (8+2,\i+0.1) -- (8+2,\i+0.2) -- (8+\i*0.1-0.1, \i+0.2) -- (8+\i*0.1-0.1, 10);
    }

    \draw[lipicsYellow] (3,0) -- (3,5)-- (1.95,5) -- (1.95,7.5)-- (3,7.5) -- (3,10);
    \draw[lipicsYellow] (11,0) -- (11,4)-- (9.95,4) -- (9.95,4.5) -- (11,4.5) -- (11,6) -- (9.95,6) -- (9.95,7.5) -- (11,7.5) -- (11,10);

    \draw (-5,0)-- (15,0) ;
    \draw (-5,10)-- (15,10) ;

    \node[draw = none] () at (-3.4,10.5) {$K_{4,4,4}$};
    \node[draw = none] () at (0.3,10.5) {$C_1$};
    \node[draw = none] () at (4.6,10.5) {$L_1$};
    \node[draw = none] () at (8.3,10.5) {$C_2$};
    \node[draw = none] () at (12.6,10.5) {$L_2$};
    \node[draw = none] () at (3, 10.3) {$s_1$};
    \node[draw = none] () at (11, 10.3) {$s_2$};

\end{tikzpicture}
    \caption{A string representation of $G$ and its $3$-subcoloring, when the instance of \textsc{No-Rainbow 3-Coloring} is composed of an hypergraph $H$ with $7$ vertices $\{v_1,\dots v_7\}$ and two hyperedges $\{v_1,v_2,v_3\}$ and $\{v_1,v_2,v_4\}$, and $A_1=\{v_5\}$, $A_2=\{v_6\}$, $A_3= \{v_7\}$. For readability, the 6 vertices additional vertices $c_{j,n+1},\dots, c_{j,n+6}$ are not represented at each layer. }
    \label{fig:StringRepresentation}
\end{figure}

Assume that $G$ has a $3$-subcoloring $\mu$. By Lemma~\ref{lemma:K444}, the vertex sets $B_1$, $B_2$, and $B_3$ must each be monochromatic with three distinct colors. \emph{W.l.o.g.}, assume that for each $k \in \{1,2,3\}$ and every $v \in B_k$, we have $\mu(v) = k$. Now, for each $k \in \{1,2,3\}$, take $v_i \in A_k$. By construction, the vertex $c_{1,i}$ is adjacent to every vertex in $B_\ell$ for $\ell \in \{1,2,3\} \setminus \{k\}$. Therefore, $\mu(c_{1,i}) = k$, as any other color assignment would form a monochromatic $P_3$ involving $c_{1,i}$ and vertices from $B_\ell$. Applying a similar argument, we find that $\mu(c_{1,n+1}) = \mu(c_{1,n+2}) = 1$, $\mu(c_{1,n+3}) = \mu(c_{1,n+4}) = 2$, and $\mu(c_{1,n+5}) = \mu(c_{1,n+6}) = 3$. In the following claim, we detail how these colors propagate throughout the graph $G$.

\begin{claim}\label{Claim:Propagation}
    For any $1\leqslant j \leqslant m$ and $1\leqslant i \leqslant n$, $\mu(c_{j,i}) = \mu(c_{1,i})$.
\end{claim}

\begin{proof}
    We prove the result by induction on $j$. The base case holds for $j=1$.

    Assume, for contradiction, that $\mu(c_{j,i}) = \mu(l_{j,i}) = k$, and \emph{w.l.o.g.}, let $k=1$. This implies that no other vertex in $C_j$ or $L_j$ can have color $1$, which leads to a contradiction because, by the induction hypothesis, $C_j$ must contain at least two vertices colored $1$, namely $c_{j,n+1}$ and $c_{j,n+2}$. Thus, by symmetry, we conclude that $\mu(c_{j,i}) \neq \mu(l_{j,i})$ and $\mu(c_{j,i}) \neq \mu(l_{j,i}')$.
    
    Now, assume for contradiction that $\mu(l_{j,i}) = \mu(l_{j,i}') = k'$, and \emph{w.l.o.g.}, let $k'=2$ (with $k=1$ as before). This assumption implies that all vertices in $L_j \setminus \{l_{j,i}, l_{j,i}'\}$ cannot have color $2$. Consider indices $i_1 \in \{n+1, n+2\} \setminus \{i\}$ and $i_3 \in \{n+5, n+6\} \setminus \{i\}$. By the induction hypothesis, $\mu(c_{j,i_\ell}) = \ell$ for $\ell \in \{1,3\}$. For both $\ell \in \{1,3\}$, the vertices $l_{j,i_\ell}$ and $l_{j,i_\ell}'$ cannot be colored $2$ or $\ell$, as that would create a monochromatic $P_3$. Consequently, we must have $\mu(l_{j,i_1}) = \mu(l_{j,i_1}') = 3$ and $\mu(l_{j,i_3}) = \mu(l_{j,i_3}') = 1$. 
    
    Finally, assigning any color to a remaining vertex in $L_j \setminus \{l_{j,i}, l_{j,i}', l_{j,i_1}, l_{j,i_1}', l_{j,i_3}, l_{j,i_3}'\}$ would again create a monochromatic $P_3$, which contradicts the subcoloring property of $\mu$. Therefore, we conclude that $\mu(c_{j,i}) \neq \mu(l_{j,i}) \neq \mu(l_{j,i}')$.
    
    By symmetry, we also obtain that $\mu(c_{j+1,i}) \neq \mu(l_{j,i})$ and $\mu(c_{j+1,i}) \neq \mu(l_{j,i}')$, which implies $\mu(c_{j,i}) = \mu(c_{j+1,i}) = k$.
    
    Notice that the same proof shows that for any $1 \leqslant j \leqslant m$, $\mu(c_{j,i}) \neq \mu(l_{j,i}) \neq \mu(l_{j,i}')$.
\end{proof}

Then, we prove that the coloring of the selector vertices ensures that their neighborhood is not a rainbow.

\begin{claim}\label{Claim:NoRainbow}
    For any $1 \leqslant j \leqslant m$, $\mu(N(s_j) \cap C_j) \neq \{1,2,3\}$, i.e., $N(s_j) \cap C_j$ is not a rainbow.
\end{claim}

\begin{proof}
    Suppose, for contradiction, that $\mu(N[s_j] \cap C_j) = \{1,2,3\}$. Then, there exists a vertex $v \in N(s_j) \cap C_j$ such that $\mu(s_j) = \mu(v) = k$. For $k = 1$, note that $v \neq c_{j,n+1}$ because $N(s_j) \subseteq C_j \setminus \{c_{j,n+1}, \dots, c_{j,n+6}\}$. Then, the path $s_j v c_{j,n+1}$ induces a monochromatic $P_3$ of color $1$, which is a contradiction. A similar argument applies for $k = 2,3$.
\end{proof}

Let $\nu : V(H) \rightarrow \{1,2,3\}$ be defined by $\nu(v_i) = \mu(c_{1,i})$ for $1 \leqslant i \leqslant n$, and we show that $\nu$ is a no-rainbow $3$-coloring of $H$. First, for any $k \in \{1,2,3\}$ and any $v_i \in A_k$, we have $\mu(c_{1,i}) = k$, and thus $\nu(v_i) = k$. Next, let $e_j = \{v_{i_1}, v_{i_2}, v_{i_3}\}$ be a hyperedge of $H$. By Claim~\ref{Claim:NoRainbow}, $\mu(N(s_j) \cap C_j) \neq \{1,2,3\}$. By construction, $N(s_j) \cap C_j = \{c_{j,i_1}, c_{j,i_2}, c_{j,i_3}\}$, so $\mu(\{c_{j,i_1}, c_{j,i_2}, c_{j,i_3}\}) \neq \{1,2,3\}$. 

Using Claim~\ref{Claim:Propagation} and the definition of $\nu$, we have
\[
\mu(\{c_{j,i_1}, c_{j,i_2}, c_{j,i_3}\}) = \mu(\{c_{1,i_1}, c_{1,i_2}, c_{1,i_3}\}) = \nu(\{v_{i_1}, v_{i_2}, v_{i_3}\}).
\]
Thus, $\nu(e_j) \neq \{1,2,3\}$, proving that $\nu$ is a no-rainbow $3$-coloring of $H$ such that for any $k \in \{1,2,3\}$ and $v \in A_k$, $\nu(v) = k$.\newline

Conversely, assume that $H$ has a no-rainbow $3$-coloring $\nu$ such that for any $k \in \{1,2,3\}$ and $v \in A_k$, $\nu(v) = k$. We construct a $3$-subcoloring $\mu$ of $G$ as follows:

\begin{itemize}
    \item For each $k \in \{1,2,3\}$ and any vertex $b \in B_k$, set $\mu(b) = k$.
    \item For any $1 \leqslant j \leqslant m$ and $1 \leqslant i \leqslant n$, set $\mu(c_{j,i}) = \nu(v_i)$. Let $k, k'$ (with $k < k'$) be the two colors different from $\nu(v_i)$; then set $\mu(l_{j,i}) = k$ and $\mu(l_{j,i}') = k'$.
    \item For any $1 \leqslant j \leqslant m$, set $\mu(c_{j,n+1}) = \mu(c_{j,n+2}) = 1$, $\mu(c_{j,n+3}) = \mu(c_{j,n+4}) = 2$, and $\mu(c_{j,n+5}) = \mu(c_{j,n+6}) = 3$. For each $i \in \{n+1, \dots, n+6\}$, let $k, k'$ (with $k < k'$) be the two colors different from $\mu(c_{j,i})$; then set $\mu(l_{j,i}) = k$ and $\mu(l_{j,i}') = k'$.
    \item For each $1 \leqslant j \leqslant n$, since $\nu$ is a no-rainbow coloring of $H$, there exists a color $k \in \{1,2,3\} \setminus \nu(e_j)$. Set $\mu(s_j) = k$.
\end{itemize}

We prove that $\mu$ is indeed a $3$-subcoloring of $G$. To establish this, we need to show that no vertex in $G$ is the center of a monochromatic $P_3$. 

First, no vertex in $B_1 \cup B_2 \cup B_3$ has a neighbor of the same color. Assume, for contradiction, that there exists a vertex $ b \in B_k $ (for some $ k \in \{1,2,3\} $) that has a neighbor with the same color. This neighbor must belong to $ C_1 $. However, if any vertex $ v_i \in C_1 $ is connected to $ B_1 \cup B_2 \cup B_3 $ and has color $ k $, it cannot be connected to $ B_k $.

Then, for any $ 1 \leqslant j \leqslant m $, notice that any vertex $ c_{j,i} \in C_j $ does not have neighbors of the same color outside $ C_j $. The only vertices it neighbors outside $ C_j $ are in the sets $ \{l_{j-1,i}, l_{j-1,i}', l_{j,i}, l_{j,i}', s_j\} $ if $ j > 1 $ and $ \{l_{j,i}, l_{j,i}', s_j\} $ if $ j = 1 $. By definition of $ \mu $, all these vertices have colors different from $ \mu(c_{j,i}) $. Additionally, since $ C_j $ induces a clique, $ c_{j,i} $ cannot be the center of a monochromatic $P_3$.

Similarly, for any $ l_{i,j} \in L_j $, all its neighbors outside $ L_j $ have different colors, except for $ s_j $ which can have the same color. Thus, if $ l_{i,j} $ is the center of a monochromatic $P_3$ represented as $ ul_{i,j}v $, then $ \{u,v\} \subseteq L_j \cup \{s_j\} $. Since $ u $ and $ v $ are not adjacent, it follows that there exists $ 1 \leqslant i' \leqslant n + 6 $ such that $ \{u,v\} = \{l_{j,i'}, l_{j,i'}'\} $; otherwise, there would be an edge between $ u $ and $ v $. However, this leads to the conclusion that $ \mu(l_{j,i'}) = \mu(l_{j,i'}') $, which contradicts the definition of $\mu$.

In conclusion, we have shown that $ \mu $ is indeed a $3$-subcoloring of $G$, completing the proof.
\end{proof}

The result can be easily generalized to \textsc{$k$-Subcoloring} for $k\geqslant 3$ by induction. Indeed, if $G$ is a co-comparability graph, then the graph $G'$ obtained by taking two disjoint copies of $G$ and adding a universal vertex is also a co-comparability graph such that $\cs(G')=\cs(G)+1$. Thus, the following holds.
\begin{theorem}
\textsc{$k$-Subcoloring} is NP-complete on co-comparability graphs for any $k\geqslant 3$.
\end{theorem}

The only unresolved case from the original question by Broersma et al. concerns the \textsc{2-Subcoloring} problem on co-comparability graphs. 
It seems plausible that this problem could be solved in polynomial time, given that \textsc{$k$-Subcoloring} is already known to be polynomially solvable for $k=2$ on chordal graphs \cite{Stacho08} but NP-complete for $k \geqslant 3$. Moreover, co-comparability graphs exhibit a certain lattice structure on their maximal cliques, raising the question of whether the algorithm for chordal graphs can be adapted to handle co-comparability graphs.

\section{The case of disk graphs}\label{sec:PositiveDisk}

Disk graphs are a natural generalization of interval graphs to two dimensions, making it appealing to adapt methods originally developed for interval graphs. As mentioned in the introduction, the best known factor~\cite{RaBrSrRa10} for polynomial-time approximation in interval graphs is $3$, and the authors asked whether a similar approach could be used for disk graphs. We first give evidence that this cannot be the case.

To do so, let us give a high-level description of their algorithm. For $k \geqslant 2$, let $BC(k)$ be the graph constructed inductively by adding a universal vertex to two disjoint copies of $BC(k-1)$, with $BC(1)$ being a single vertex. A straightforward proof shows that $\chi_s(BC(k)) \geqslant k$ for any $k \geqslant 1$, and observe that $BC(k)$ are indeed interval graphs, for $k \geqslant 1$. Then, an interval is said \emph{internal} if it properly contains two independent intervals, and \emph{external} otherwise. Given a graph $G$, let $V_1,...,V_k$ be the partition of the vertex set $V(G)$ obtained iteratively, for $i=1, \dots, k$, by defining $V_i$ as the set of external vertices of the graph induced by $V(G) \setminus \cup_{j < i} V_j$. Observe that such a partition implies the existence of an induced subgraph isomorphic to $BC(k)$. Additionally, it can be shown that each $G[V_i]$ has a $3$-subcoloring for all $1\leqslant i \leqslant k$, leading to a $3$-approximation.

A key ingredient of this proof relies on the fact that an interval graph where no interval is internal can be subcolored with a constant number of colors. We show that this no longer holds in disk graphs. We even prove the stronger result that for any $k \geqslant 1$, $BC(k)$ is a \textit{proper disk graph}, that is an intersection graph of disks such that no disk properly contains another one.

\begin{figure}[!h]
\centering
\begin{tikzpicture}[scale=0.8]
\clip rectangle (0,0) rectangle (10,9.5);

\draw[dashed] (4.5,0) -- (4.5,9.5);
\draw[dashed] (5,0) -- (5,9.5);
\draw[dashed] (5.5,0) -- (5.5,9.5);

\draw (4,2.5) circle (2) ;
\draw (4,7) circle (2) ;

\draw[purple, fill = purple, fill opacity = 0.2] (-25,4.75) circle (30.35); 

\node[draw = none] () at (6.8, 2.5) {\LARGE $G_1$}; 
\node[draw = none] () at (6.8, 7) {\LARGE $G_2$}; 

\node[draw = none] () at (1, 4.75) {\color{purple} \LARGE $D$}; 

\end{tikzpicture}
\caption{The proper disk representation of $BC(k)$. The dashed lines, from left to right, represent $L_1$, $L_2$ and $L_3$ respectively. $G_1$ and $G_2$ are two proper disk representations of $BC(k-1)$, where each disk intersects both $L_1$ and $L_3$.}\label{fig:BCkProper}
\end{figure}
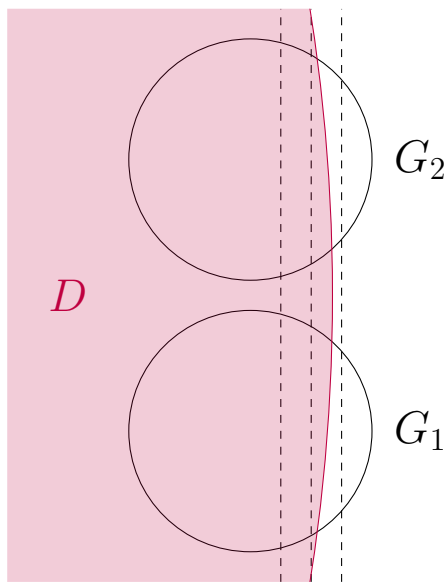
 
\begin{theorem}\label{thm:BCkProper}
 For any $k\geqslant 1$, the graph $BC(k)$ is a proper disk graph.
\end{theorem}

\begin{proof}
Consider two parallel lines $L_1$ and $L_2$ at a distance $1$ from each other. We prove, by induction on $k$, that there exists a proper disk representation of $BC(k)$ such that each disk intersects both lines $L_1$ and $L_2$.

For $k = 1$, $BC(1)$ is a single vertex and has the desired property since it can be represented by a disk of radius $1$ that intersects both lines.

Now, assume that $BC(k-1)$ has the desired property. Let $L_3$ be a line parallel to $L_1$ and $L_2$, at a distance of $1$ from $L_2$ and $2$ from $L_1$. Consider two copies, $G_1$ and $G_2$, of $BC(k-1)$, and by rescaling, each graph has a proper disk representation such that every disk intersects the lines $L_1$ and $L_3$. We place these two graphs next to each other, as illustrated in Figure~\ref{fig:BCkProper}, so that the disk representations of $G_1$ and $G_2$ are disjoint and all disks intersect both $L_1$ and $L_3$.

Finally, add a disk $D$ with a radius large enough such that:
\begin{itemize}
    \item $D$ intersects both $L_1$ and $L_2$, but not $L_3$.
    \item For any $v \in V(G_1) \cup V(G_2)$, we have $D_v \cap L_2 \subseteq D$, i.e., $D$ covers any point of $L_2$ that is covered by a disk in the representation of $G_1$ or $G_2$.
\end{itemize}

Consider any disk $D'$ in the representation of $G_1$ or $G_2$. Since $D'$ intersects both $L_1$ and $L_3$, it also intersects $L_2$ and, consequently, $D$. Moreover, $D'$ is not contained within $D$, because otherwise $D$ would intersect both $L_1$ and $L_3$, which contradicts our construction.

Thus, $D$ intersects all disks in $G_1$ and $G_2$ without containing any of them, and therefore $BC(k)$ has a proper disk representation where every disk intersects both $L_1$ and $L_2$.
\end{proof}

In Section~\ref{sec:positiveunit}, we showed that any unit disk graph can be subcolored with a constant number of colors. The same result cannot be obtained for disk graphs, as there exist disk graphs (and even interval graphs) with subchromatic number logarithmic in the number of vertices~\cite{broersma2002}. We now prove Theorems~\ref{thm:UpperBoundDiskGraph} and~\ref{thm:ApproxDiskGraph}, which provide an almost tight upper bound and an approximation algorithm, respectively, in the case of disk graphs.

Our proof relies on a three-steps decomposition, where at each time we find a balanced separator which is itself decomposed, until the pieces we obtain are simple enough to be able to color them with a constant number of colors. In particular, the last step consists of a subclass of disk graphs dubbed \textit{$\Delta$-disk graphs} which turns to be included in co-comparability graphs. From an algorithmic perspective, we propose a polynomial-time $O(1)$-approximation algorithm for the subchromatic number of those $\Delta$-disk graphs.

\subsection{Decomposition of disk graphs and $\Delta$-disk graphs}

\begin{figure}
\centering
\begin{tikzpicture}
\clip rectangle (-5,-3) rectangle (5,3);
 \draw [-] (0,-3)--(0,3) ;
 \draw [-](-5,0) -- (5,0);
\foreach \x / \y in {0.40483056047342847/0.6116640952449762, 0.7365161460470414/3.217477505756438, -0.9301199401352656/1.7168326285897628, 0.5853918426255926/1.955272861393986, 1.775327341409589/-4.630882759498221, -0.9886929496758988/3.9515839154655286, -4.0929226645881585/1.209518463718592, 1.0675924884197356/1.2735167503129643, -1.6823769875793095/3.9859255533584426, 0.4725605012624058/-2.031155172749482, 0.8823411893491253/2.9066934597807133, 3.4067237394104293/-3.6358043500369117, -2.861881650038868/-1.111575368533092}{
\draw[black,opacity = 0.5] (\x,\y) circle (0.5);}
\foreach \x / \y in {-2.5621705424098074/-0.27199753201318083, -2.3465956067264253/0.3691289837869378}{
\draw[red] (\x,\y) circle (0.5);}
\foreach \x / \y in {0.18703865036567963/-0.04084962778366008, -0.3431499580292409/-0.39045587365217355}{
\draw[red, fill= red, fill opacity = 0.2] (\x,\y) circle (0.5);}
\foreach \x / \y in {4.2078410311893295/2.69138197196618, 2.2540158608786567/-2.491915030739881, 1.7943733834746145/-6.7143279188547655, 2.665965440600177/2.064299757087773, -1.7378802055020701/0.9347210117965864, -3.0270360559927605/4.49675685691714, -5.138272560051211/-0.8422564169837912, -4.753906533590628/-3.798498659360405, 4.856122791690936/5.474055869661155, 0.4336203674036645/-1.8099950572532595}{
\draw[black,opacity = 0.5] (\x,\y) circle (0.75);}
\foreach \x / \y in {1.7116314204629155/0.13378024932762475, 3.1692717468683/0.648317246235075, 2.710753040639909/-0.5952284374142889, -3.2482679829135304/0.6018167190299676, 2.250352168981366/-0.4980599798523478, 5.573030656326713/0.2776071739408765}{
\draw[red] (\x,\y) circle (0.75);}
\foreach \x / \y in {0.6649780980803982/-0.2616678420935648, 0.009340771345589301/-0.2606086918314971,  -0.6174540048924511/-0.3199668839721195}{
\draw[red, fill= red, fill opacity = 0.2] (\x,\y) circle (0.75);}
\foreach \x / \y in {-1.1755461925845097/-5.198895499638054, -1.0392062373416842/-1.5240578040612012, 2.954912697755886/-5.713945633567547, -2.345186513235602/1.3027816788685012, -0.868674458697813/1.7553077493431295, 1.194906797120306/2.5907389071684763, -4.044104271511318/-2.0614518929667196, -2.757205862755475/3.7035452042722827, 4.618721377784992/-6.5044115956001445, 7.372052577421708/5.384820995608887}{
\draw[black,opacity = 0.5] (\x,\y) circle (1);}
\foreach \x / \y in {2.3694397521501056/-0.04556953825454385, -2.1634242044375345/-0.6170133163657778, -5.193977074139313/0.42299905447682, 3.6694120409948354/0.5063958401611758, -4.173112954832848/0.04800202607786376, 1.645269692209235/-0.729905171902737, -1.908839150931836/-0.13263324892205716, -1.0422364827199548/0.15212172174720312}{
\draw[red] (\x,\y) circle (1);}

\foreach \x/\y in {1.2/1.2, -1.4/-1.45}{
\draw[red, fill= red, fill opacity = 0.2] (\x,\y) circle (1.5);}
\end{tikzpicture}

\caption{A disk graph and its decomposition. The circles with red lines form a line disk graph, which is a separator for the whole graph. In addition, the vertices filled in red can be partitioned into four $\Delta$-disk graphs and one clique.}\label{fig:DecompositionDisks}
\end{figure}
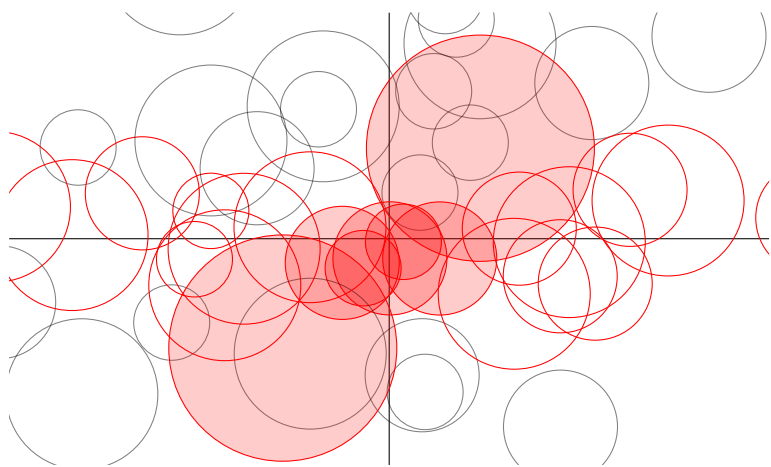

In the following, we say that a subset of vertices $S \subseteq V(G)$ is a \textit{balanced separator} if $V(G)\setminus S$ can be partitioned into two sets $A$ and $B$ containing at most $\lceil |V(G)|/2 \rceil$ vertices each, and such that there is no edge between $A$ and $B$. In addition, we use the notation $\Delta = \mathbb{R}^2_+$ for the first quadrant of the plane. We now define three subclasses of disk graphs.

\begin{definition}
A disk graph is called a linear disk graph if there exists an embedding of $G$ such that a line crosses all disks of $G$. 


A disk graph $G$ is called a $\Delta$-disk graph if there exists a disk representation where each disk intersects both axis but does not contain the point $(0,0)$, or more formally such that for any $v\in V(G)$, we have $x_v > 0$, $y_v > 0$ and $\max(x_v,y_v)\leqslant r_v < \sqrt{x_v^2+y_v^2}$.
\end{definition}

For the remainder of this section, we assume that the representation of any $\Delta$-disk graph satisfies the properties given in the above definition. Additionally, for a given $\Delta$-disk graph $G$ with its associated disk representation and a vertex $v\in V(G)$, we denote by $d_v = \sqrt{x_v^2 + y_v^2}$ the distance of the point $(x_v, y_v)$ from the origin.

\begin{theorem}[Decomposition of disk graph]\label{thm:DecompositionDisk}
Let $G$ be a disk graph.
\begin{enumerate}
\item \label{label:first} Any disk graph contains a balanced separator which induces a linear disk graph.

\item \label{label:second} Any linear disk graph contains a balanced separator $S$ which can be partitioned into $(V_i)_{1\leqslant i\leq 5}$ such that $G[V_i]$ is a $\Delta$-disk graph for $1\leqslant i \leqslant 4$ and $G[V_5]$ is a clique.
\end{enumerate}
\end{theorem}

\begin{proof}
Proof of \ref{label:first}. Consider a representation of the disk graph $G$, and let $y$ be the median of all the ordinates of the disks. Consider the line $L$ parallel to the abscissa axis crossing the point $(0,y)$, and let $S$ be all disks that cross $L$. Observe by construction that all disks of $V \setminus S$ with ordinate less than $y$ do not intersect those with ordinate greater than $y$, and both sets represent at most $\lceil n/2 \rceil$ disks.

Proof of \ref{label:second}. If $G$ is a linear disk graph, w.l.o.g. we consider that this line is exactly the abscissa axis. Let $x$ be the median of all abscissa of the vertices, and consider the line $L'$ perpendicular to $L$ crossing the point $(x,0)$. By similar arguments as previously, observe that the set of disks crossing $L'$ is a balanced separator which induces a disk graph where each disk intersect both axis. Let $V_5$ be the set of vertices such that the corresponding disks intersect the intersection of both lines, and let $V_1,V_2,V_3,V_4$ be the partition of the remaining vertices corresponding to the division of the plane into four subspaces formed by the two perpendicular lines. By definition, $G[V_i]$ is a $\Delta$-disk graph for $1\leqslant i\leqslant 4$, and since each disk of $V_5$ intersects the same point, $G[V_5]$ is a clique.
\end{proof}

An illustration of the decomposition is given on Figure~\ref{fig:DecompositionDisks}.

\subparagraph{Relation between $\Delta$-disk graphs and other classes} Due to the constraints on their structure, $\Delta$-disk graphs are far less general than classical disk graphs. In particular, they enjoy many structural properties that we can formalize using other known graph classes. For instance, we can show that they strictly generalize interval graphs, and that they are co-comparability graphs.

\begin{theorem}\label{thm:DeltaInterval}
Any interval graph is a $\Delta$-disk graph. Moreover, there exists a $\Delta$-disk graph that is not an interval graph.
\end{theorem}

\begin{proof}
We say that two interval representations $ \{(\ell_v, r_v)\}_{v \in V(G)}$ and $ \{(\ell_v', r_v')\}_{v \in V(G)}$ of an interval graph $G$ \emph{coincide} if mapping $\ell_v$ to $\ell_v'$ and $r_v$ to $r_v'$ for all $v \in V(G)$ preserves the order on the real line.

\begin{claim}\label{Claim:CenteredDisks}
For any $t > 0$, let $D$ be the disk centered at $(t, t)$ with radius $t$. Let $a_1$ and $a_2$ be the smallest and largest values of $a > 0$ such that $(a, a) \in D$. Then,
\[
a_1 = \left( 1 - \frac{1}{\sqrt{2}} \right) t \quad \text{and} \quad a_2 = \left( 1 + \frac{1}{\sqrt{2}} \right) t
\]
\end{claim}

\begin{proof}
Notice that $a_1$ and $a_2$ are the positive solutions of the equations $t^2 = 2(t - a_1)^2$ and $t^2 = 2(a_2 - t)^2$, respectively.
\end{proof}

Let $G$ be an interval graph, and consider an interval representation $\{(\ell_v, r_v)\}_{v \in V(G)}$ of $G$ where intervals are sorted with respect to their right endpoints.
Define $G_i = G[\{v_1, \dots, v_i\}]$ for $1 \leqslant i \leqslant n$. We construct a disk representation of $G_i$ by induction:

\begin{itemize}
    \item Add a disk $D_1$ with center $(1, 1)$ and radius $1$ corresponding to $v_1$. Let $\ell_1'$ and $r_1'$ be the smallest and largest $t > 0$ such that $(t, t) \in D_1$.
    
    \item Assume that for $G_{i-1}$, a disk representation has been constructed such that:
    \begin{itemize}
        \item For any $1 \leqslant j < i$, the center of each disk $D_j$ is at $(t, t)$ for some $t > 0$, and each $D_j$ intersects both axes but not the origin.
        
        \item The set $\{(\ell_j', r_j')\}_{1 \leqslant j < i}$ forms an interval representation of $G_{i-1}$, coinciding with $\{(\ell_j, r_j)\}_{1 \leqslant j < i}$.
    \end{itemize}
    
    Now consider $G_i$ and distinguish two cases:
    \begin{itemize}
        \item If $\ell_i > r_{i-1}$, let $t_i = \frac{1}{1 - \frac{1}{\sqrt{2}}}(r_{i-1}' + 1)$ and let $D_i$ be the disk with center $(t_i, t_i)$ and radius $t_i$. By Claim~\ref{Claim:CenteredDisks}, the values $\ell_i' = r_{i-1}' + 1$ and $r_i' = \frac{\sqrt{2} + 1}{\sqrt{2} - 1}(r_{i-1}' + 1)$ are the smallest and largest real numbers $a > 0$ such that $(a, a) \in D_i$. This shows that $\{(\ell_j', r_j')\}_{1 \leqslant j \leqslant i}$ is an interval representation of $G_i$.
        
        \item If $\ell_i < r_{i-1}$, let $e_1, e_2 \in \{\ell_1, r_1, \dots, \ell_{i-1}, r_{i-1}\}$ be two events such that $e_1 < \ell_i < e_2$ and $|e_2 - e_1|$ is minimized. Let $e_1', e_2' \in \{\ell_1', r_1', \dots, \ell_{i-1}', r_{i-1}'\}$ be the values corresponding to $e_1$ and $e_2$ in the new interval representation of $G_{i-1}$. Choose $\ell_i' > 0$ such that $e_1' < \ell_i' < e_2'$, and set $t = \max\left( \frac{1}{1 - \frac{1}{\sqrt{2}}} \ell_i', r_{i-1}' \right)$. Define $D_i$ as the disk centered at $(t, t)$ with radius $R = \sqrt{2}(t - \ell_i')$. This ensures that:
        \begin{itemize}
            \item The smallest $a > 0$ such that $(a, a) \in D_i$ satisfies $2(t-a)^2 =R^2$, and thus is exactly~$\ell_i'$.
            \item The largest value $a > 0$ such that $(a, a) \in D_i$ satisfies $a > r_{i-1}'$, as $t > r_{i-1}'$. Set $r_{i}'$ to be this largest value.
        \end{itemize}
        Altogether, $\{(\ell_j', r_j')\}_{1 \leqslant j \leqslant i}$ forms an interval representation of $G_i$ which coincides with $\{(\ell_j, r_j)\}_{1 \leqslant j \leqslant i}$, and all disks $D_1, \dots, D_i$ intersect both axes but not the origin.
    \end{itemize}
\end{itemize}

At the end, we have $G_n = G$ with a disk representation that satisfies the properties of a $\Delta$-disk graph. Finally, note that $C_4$, the cycle on four vertices, is a $\Delta$-disk graph (see Figure~\ref{fig:C4DeltaDisk}) but not an interval graph. 

\begin{figure}[h]
\centering
\begin{tikzpicture}[scale=0.7]
\draw [->, thick] (0,-0.25)--(0,8) ;
\draw [->, thick](-0.25,0) -- (10,0);

\clip rectangle (-1,-1) rectangle (10,8);
\draw[fill = gray, fill opacity = 0.2] (2,1) circle (2);
\draw[fill = gray, fill opacity = 0.2] (1,2) circle (2);

\draw[fill = gray, fill opacity = 0.2] (50,100) circle (107.8);
\draw[fill = gray, fill opacity = 0.2] (100,50) circle (107.8);

\end{tikzpicture}
\caption{The cycle $C_4$ is a $\Delta$-disk graph.}\label{fig:C4DeltaDisk}
\end{figure}
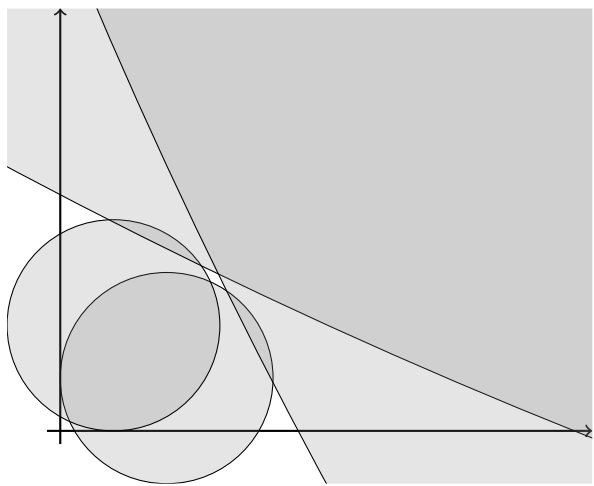
\end{proof}

\begin{theorem}\label{thm:Cocomp}
Any $\Delta$-disk graph $G$ is a co-comparability graph with respect to the following co-comparability order: for any two vertices $u, v \in V(G)$, define $u \prec v$ if and only if $x_u < x_v$, $y_u < y_v$, and $D_u \cap D_v = \emptyset$.
\end{theorem}

\begin{proof}
Observe that in a $\Delta$-disk graph $G$, if two disks $D_u$ and $D_v$ ($u, v \in V(G)$) do not intersect and $x_u < x_v$, then $y_u < y_v$. Indeed, suppose for contradiction that  $y_u \geqslant y_v$. In this case, both $D_u$ and $D_v$ intersect at the point $(x_u, y_v)$, since the radius of $D_u$ is at least $x_u$ and the radius of $D_v$ is at least $y_v$, contradicting the assumption that $D_u$ and $D_v$ are disjoint.
Finally the relation is indeed transitive:  if $u \prec v$ and $v \prec w$, then $u \prec w$ since $D_v$ separates $\mathbb{R}_+^2$.  
\end{proof}

\subparagraph{Properties of $\Delta$-disk graphs.} We first establish some basic properties of $\Delta$-disk graphs that will be useful in our approximation algorithm. The following lemma, which will be applied frequently (often implicitly), states that if two disks in a $\Delta$-disk graph intersect, then their intersection must occur within the region $\Delta =  \mathbb{R}_+^2$. In other words, there cannot be a ``hidden'' intersection point located outside $\Delta$, e.g. at a point with a negative coordinates. This result follows directly from the geometry of the disks.

\begin{lemma}
Let $G$ be a $\Delta$-disk graph and let $uv \in E(G)$. Then, the disks $D_u$ and $D_v$ intersect within $\Delta$; that is, the intersection $D_u \cap D_v \cap \Delta$ is non-empty.
\end{lemma}

\begin{proof}
Observe that the line segment joining the centers of $D_u$ and $D_v$ lies entirely within $\Delta$. Since the disks intersect, at least one point on this segment must be in the intersection $D_u \cap D_v$, and this point also lies within $\Delta$.
\end{proof}

\begin{lemma}\label{lemma:BoundNonEdge}
Let $G$ be a $\Delta$-disk graph, and let $u, v \in V(G)$ such that $uv \notin E(G)$ and $d_u < d_v$. Then, $\min(x_v, y_v) \geqslant \max(x_u, y_u) + r_u$.
\end{lemma}

\begin{proof}
\emph{W.l.o.g.}, assume that $x_u \geq y_u$. We aim to show that if $\min(x_v, y_v) < x_u + r_u$, then $u$ and $v$ are adjacent in $G$. 

Since $d_u < d_v$, vertex $u$ precedes $v$ in the co-comparability ordering. Thus, by Theorem~\ref{thm:Cocomp}, we have $x_u < x_v$ and $y_u < y_v$. We now consider two cases:

\begin{itemize}
    \item If $x_v \leq x_u + r_u$. In this case, observe that both disks $D_u$ and $D_v$ intersect at the point $(x_v, y_u)$, making $u$ and $v$ adjacent in $G$, which contradicts the assumption that $uv \notin E(G)$.

    \item If $x_v > x_u + r_u$ but $y_v < x_u + r_u$. We show that $D_u$ and $D_v$ still intersect. We denote $C_u =(x_u,y_u)$ and $C_v=(x_v,y_v)$, and applying the triangle inequality gives:
    \[
    d(C_u, C_v) \leq (x_v - x_u) + (y_v - y_u) < (x_v - x_u) + (x_u + r_u - y_u).
    \]
    Since $x_v \leq r_v$ by assumption, it follows that $d(C_u, C_v) < r_u + r_v$, implying that $D_u$ and $D_v$ intersect, which again contradicts $uv \notin E(G)$.
\end{itemize}
Thus, we conclude that $\min(x_v, y_v) \geq \max(x_u, y_u) + r_u$.
\end{proof}

\begin{lemma}\label{lemma:Convexity}
Let $G$ be a $\Delta$-disk graph and $v\in V(G)$. If, for some $t\leqslant \min(x_v,y_v)$, $D_v\cap [0;t]^2$ is not empty, then $D_v$ intersects the point $(t,t)$.
\end{lemma}

\begin{proof}
Let $(x,y) \in D_v\cap [0;t]^2$. We notice that $x_v-x \geqslant x-t$ since $x\leqslant t$, and similarly $y_v-y\geqslant y-t$. It follows that $(x_v,y_v)$ is closer to $(t,t)$ than to $(x,y)$, and thus $D_v$ intersects $(t,t)$ as well. 
\end{proof}

\subsection{Polylogarithmic upper bound}

\begin{figure}[t]
\centering
\begin{tikzpicture}
\clip rectangle (-1,-1) rectangle (10.5,8.5);

\draw [->] (0,-0.25)--(0,8) ;
\draw [->](-0.25,0) -- (10,0);

\draw [<->] (-0.5,0) -- (-0.5,5) ;
\draw (-0.75,2.5) node {$\alpha$} ;
\draw [<->] (0,-0.5) -- (5,-0.5) ;
\draw (2.5,-0.75) node {$\alpha$} ;

\draw [dashed, line width = 0.8] (5,5) -- (10.5,5) ;
\draw [dashed, line width = 0.8] (5,5) -- (5,10.5) ;

\draw [dashed, line width = 0.8] (0,2.5) -- (2.5,2.5) ;
\draw [dashed, line width = 0.8] (2.5,0) -- (2.5,2.5) ;

\draw [blue, line width = 0.8] (0,5) -- (5,5) ;
\draw [blue, line width = 0.8] (5,0) -- (5,5) ;

%

\draw (2.5,6.5) node {$\Delta_3$} ;
\draw (7.5,6.5) node {$\Delta_2$} ;
\draw (1.25,1.25) node {$\Delta_1$} ;

\draw (5,5) node [above right] {$X$} node{$\bullet$};
\draw (2.5,2.5) node [above right] {$X'$} node{$\bullet$};
\draw (0,0) node [below left] {$O$} node{$\bullet$};

\draw [black, fill =gray, fill opacity =0.2] (1, 1) circle (1.1) ;
\draw [black, fill =gray, fill opacity =0.2] (1.3, 1.2) circle (1.3) ;
\draw [black, fill =gray, fill opacity =0.2] (20, 20) circle (20) ;

\draw [red, fill =red, fill opacity =0.0] (1.5, 3.7) circle (3.8) ;
\draw [red, fill =red, fill opacity =0.0] (2.5, 3.5) circle (3.8) ;
\draw [red, fill =red, fill opacity =0.0] (3.5, 3.5) circle (3.6) ;
\draw [red, fill =red, fill opacity =0.0] (3.5, 3) circle (3.6) ;

\end{tikzpicture}
\caption{A decomposition of a $\Delta$-disk graph. All the red disks have centers in $\Delta_3$ and thus intersect $X'$. They gray disks are those with center in $\Delta_1$ (resp. $\Delta_2$) and that do not intersect $X'$ (resp. $X$). The blue lines represent the separator of Claim \ref{Claim2}.}\label{fig:DeltaDisks}
\end{figure}
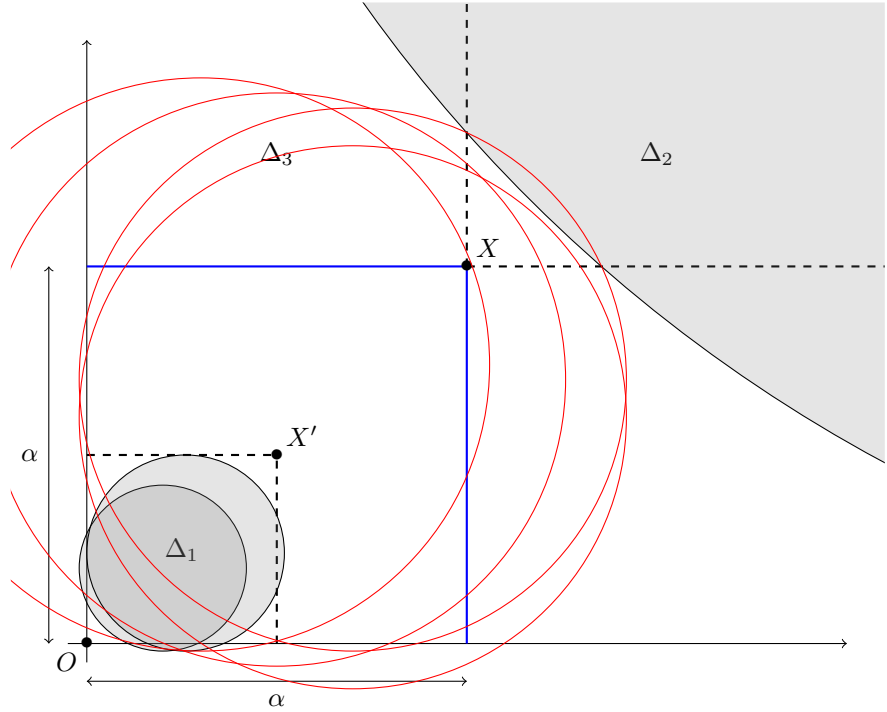
 
Here we prove that $\Delta$-disk graphs have logarithmic subchromatic number.

\begin{lemma}
Any $\Delta$-disk graph contains a balanced separator which can be partitioned into two cliques.

\end{lemma}
\begin{proof} Let $G$ be a $\Delta$-disk graph, and let $\mathcal{D}_G = \{D_v = D((x_v,y_v), r_v) \mid v\in V(G) \}$ be a disk representation of $G$.
Let $\alpha$ be the median of the abscissa of all the centers in the representation $\mathcal{D}_G$, and let $\Delta = \{(x, y) \mid x > 0 \wedge y > 0 \}$ be the set of points with positive coordinates. Consider a partition of $\Delta$  using the following sets:
 \begin{itemize}
 \item $\Delta_1 = \{(x,y) \in \mathbb{R}^2 \mid 0< x < \alpha/2  \wedge 0< y < \alpha/2 \}$ 
\item $\Delta_2 = \{(x,y) \in \mathbb{R}^2 \mid \alpha < x \wedge \alpha < y \}$ 
\item $\Delta_3 = \Delta \backslash (\Delta_1\cup \Delta_2)$
 \end{itemize}
We can notice that $\Delta =  \bigsqcup_{i=1}^3 \Delta_i$. The partition is depicted in Figure \ref{fig:DeltaDisks}. In addition, let $X=(\alpha,\alpha)$ and  $X' = (\alpha/2, \alpha/2)$.

\begin{claim}\label{Claim1}
Any disk of $\mathcal{D}_G$ with center in $\Delta_3$ contains $X'$.
\end{claim}

\begin{proof}
Let $v\in V(G)$ such that $C_v=(x_v,y_v)\in \Delta_3$ and $x_v\leqslant \alpha$.  Consider the point $A$ of coordinates $(x_v, \alpha/2)$. We have $d(C_v,X') \leqslant d(C_v,A) + d(A,X') \leqslant y_v-\alpha/2 + \alpha/2 = y_v$. Since $r_v \geq y_v$, $D_v$ intersects $X'$. By symmetry, all  the remaining disks with center in $\Delta_3$ intersect $X'$.
\end{proof}

\begin{claim}\label{Claim2}
Let $v\in V(G)$ such that  $C_v = (x_v,y_v)\in \Delta_1$ and $X' \notin D_v$, and let $u \in V(G)$ such that $C_u = (x_u,y_u)\in \Delta_2$ and $X\notin D_u$. Then $D_v$ and $D_u$ do not intersect.
\end{claim}

\begin{proof}

Let $A$ be a point of coordinates $(\alpha, y_A)$ with $y_A \leqslant \alpha$ ($A$ is a point below $X$ on the blue line in Figure \ref{fig:DeltaDisks}). We show that $A$ intersects neither $D_v$ nor $D_u$. First, we have
\begin{align*}
d(C_u,A)^2 &= (x_u-\alpha)^2 + (y_u - y_A)^2 \\
&\geqslant (x_u-\alpha)^2 + (y_u - \alpha)^2 \\
&= d(C_u,X)^2
\end{align*} 
However, $D_u$ does not contain $X$ and so $d(C_u,X) > r_u$. It follows that $D_u$ does not contain $A$. Then, notice that $d(C_v,A)^2 \geqslant (\alpha -x_v)^2 > (\alpha/2)^2$. Now, we want to show that $r_v \leqslant \alpha/2$. Using $d(C_v,(0,0)) >r_v$ and $d(C_v,X') >r_v$, we obtain 
\[
\left\lbrace \begin{array}{l}
x_v^2+y_v^2 > r_v^2 \\
\left(\frac{\alpha}{2}-x_v\right)^2 + \left(\frac{\alpha}{2}-y_v\right)^2 > r_v^2
\end{array}\right.
\]
By summing the two inequalities and optimizing the quadratic function $z \mapsto z^2 + (\alpha/2-z)^2$, we obtain that $ 2r_v^2 < 4(\alpha/4)^2 $ and thus $r_v < \frac{\alpha}{2\sqrt{2}} < \frac{\alpha}{2}$.
Similar arguments allow to consider any point of coordinate ($x_A, \alpha)$ with $x_A \leqslant \alpha$. The two semi-lines are separating the plane, and thus $D_v$ and $D_u$ do not intersect.
\end{proof}

Let $V_1$ (resp.\ $V_2$) be the vertices with disk center in $\Delta_1$ (resp.\ $\Delta_2$) that do not contain $X'$ (resp.\ $X$). By definition of the median, and since $\Delta_2$ do not contain any center of abscissa~$\alpha$, $|V_1|<\lfloor n/2 \rfloor$ and $|V_2|< \lfloor n/2 \rfloor$. Then, by Claim \ref{Claim2}, there is no edge between $V_1$ and $V_2$. Then, consider $V_3$ to be the set of vertices with disk center in $\Delta_3$ and those with center in $\Delta_1$ which intersect $X'$. By Claim \ref{Claim1}, they all intersect $X'$ and thus $G[V_3]$ is a clique. In a same way, let $V_4$ be the set of vertices with disk centers in $\Delta_2$ which intersect $X$. By definition, $G[V_4]$ is a clique. Hence, $V_3 \cup V_4$ is indeed a balanced separator which can be partitioned into two cliques.
\end{proof}

By coloring the balanced separator with two colors and inductively coloring the remaining vertices, we obtain the following.
\begin{lemma}
For any $\Delta$-disk graph $G$, $\cs(G) \leq 2\log_2(n)$.
\end{lemma}

\begin{proof}
We show the result by induction. Let $C(n)$ be the maximum subchromatic number of a $\Delta$-disk graph with $n$ vertices. Using the previous lemma, we have $C(n) \leqslant C( \lfloor n/2 \rfloor) + 2 $ and by induction $C(n) \leqslant 2\log_2(n)+1$ for any $n\geqslant 1$.
\end{proof}

The proof of Theorem~\ref{thm:UpperBoundDiskGraph} is straightforward with the previous lemma and the decomposition of disk graphs of Theorem~\ref{thm:DecompositionDisk}. Indeed, through classical induction, any linear disk graph $G$ satisfies $\chi_s(G) = O(\log^2(n))$, and finally, any disk graph $G$ satisfies $\chi_s(G) = O(\log^3(n))$. Furthermore, such a subcoloring can be found in time $O(n\log^3(n))$ given the representation, using an algorithm that computes the median in linear time.
\medskip

\subsection{A $O(\log^2(n))$-approximation algorithm}

Given a $\Delta$-disk graph $G$ and two vertices $u,v\in V(G)$, we say that $v$ \emph{contains} $u$ if, and only if $N[u] \subseteq N[v]$. By extension, a vertex $v$ contains a set $U\subseteq V(G)$ if for any $u\in U$, $v$ contains $u$.

\begin{theorem}\label{thm:ApproxDeltaDisks}
There exists a constant-factor approximation algorithm for computing the subchromatic number of a $\Delta$-disk graph.
\end{theorem}

The core result for proving this theorem is the following lemma:

\begin{lemma}\label{lemma:ContainsDeltaDisk}
Let $G$ be a unit disk graph and $S=\{v_1,...,v_9\}$ an independent set of size $9$ where the $v_i$'s are sorted by increasing co-comparability order. If a vertex $v$ is adjacent to all vertices of $S$, then $v$ contains $v_5$.
\end{lemma}

\begin{proof}

For each $i \in \{1, \ldots, 9\}$, let $D_i$ denote the disk of vertex $v_i$ in the disk representation of $G$, with center $(x_i, y_i)$ and radius $r_i$. Define $t = \min\{a \mid (a, a) \in D_5\}$ and $T = \max(x_5, y_5) + r_5$.

\begin{claim}
The inequality $\max(x_v, y_v) > T$ holds.
\end{claim}
\begin{proof}
Suppose, for contradiction, that $\max(x_v, y_v) \leqslant \max(x_5, y_5) + r_5$. Then, since $r_v^2 < x_v^2 + y_v^2$, we have $r_v < \sqrt{2} \cdot \max(x_v, y_v)$. Consequently, $\max(x_v, y_v) + r_v < (1 + \sqrt{2}) \max(x_v, y_v) \leq (1 + \sqrt{2})(\max(x_5, y_5) + r_5)$. Applying Lemma~\ref{lemma:BoundNonEdge}, we obtain that $\max(x_v, y_v) + r_v \leq (1 + \sqrt{2}) \min(x_6, y_6)$.

Noting that $\min(x_6, y_6) \leq \max(x_6, y_6)$ and $\min(x_6, y_6) \leq r_6$, we find
\[
\max(x_v, y_v) + r_v < \frac{1 + \sqrt{2}}{2} (\max(x_6, y_6) + r_6).
\]
Repeating the reasoning iteratively, we conclude
\[
\max(x_v, y_v) + r_v < \frac{1 + \sqrt{2}}{4} (\max(x_7, y_7) + r_7) < \max(x_7, y_7) + r_7 \leq \min(x_8, y_8).
\]
The last inequality holds by Lemma~\ref{lemma:BoundNonEdge}. Consequently, $D_v$ does not intersect the segments $[(0, y_8), (x_8,y_8)]$ and $[(x_8, y_8), (x_8, 0)]$, which are fully contained in $D_8$. This implies that $D_v$ is separated from $D_9$, contradicting that $v$ is adjacent to $v_9$.
\end{proof}

\begin{claim}
$D_v$ intersects both points $(T, T)$ and $(t/4, t/4)$.
\end{claim}
\begin{proof}
We begin with the intersection at $(T, T)$. Consider the cases:
\begin{itemize}
    \item If $y_v > T$ and $x_v \leq T$, then $d((x_v, y_v), (T, T)) \leqslant (y_v - T) + (T - x_v) \leq y_v \leq r_v$, so $D_v$ intersects $(T, T)$.
    \item If $y_v \leqslant T$ and $x_v > T$, a similar argument shows that $D_v$ intersects $(T, T)$.
    \item If $y_v > T$ and $x_v > T$, then by Lemma~\ref{lemma:Convexity}, if $D_v$ does not intersect $(T, T)$, it does not intersect any point in the square $[0, T]^2$. In particular, it does not intersect $D_4$, contradicting that $\max(x_4, y_4) + r_4 \leq \min(x_5, y_5)$ by Lemma~\ref{lemma:BoundNonEdge}.
\end{itemize}

Now, we show that $D_v$ intersects $(t/4, t/4)$. Assume for contradiction that $D_v$ does not intersect $(t/4, t/4)$. Since $\max(x_v,y_v)>T$, we have $\max(x_v,y_v)>t/4$. If $x_v\leqslant t/4$ and $y_v>t/4$, then $d((x_v, y_v), (t/4, t/4)) \leqslant (y_v - t/4) + (t/4 - x_v) \leqslant y_v \leq r_v$, so $D_v$ intersects $(t/4, t/4)$, which is not possible. With a similar reasoning, the case $x_v>t/4$ and $y_v \leqslant t/4$ is not possible as well. The only remaining case is  $x_v > t/4$ and $y_v > t/4$. By Lemma~\ref{lemma:Convexity}, $D_v$ contains no point from $[0, t/4]^2$. We show that $v$ cannot be adjacent to $v_1$ in this case: by Lemma~\ref{lemma:BoundNonEdge}, $\max(x_1, y_1) + r_1 \leq \min(x_2, y_2)$. Reapplying this lemma, we get
\[
\max(x_1, y_1) + r_1 \leq \frac{1}{2} \min(x_3, y_3) \leq \frac{1}{4} \min(x_4, y_4) \leqslant \frac{1}{4}t
\]
where the last inequality holds since $v_4$ is before $v_5$ in the co-comparability order, and $D_5$ contains the point $(t,t)$.
It follows that $D_1 \cap \Delta \subset [0, t/4)^2$, and thus $D_v$ does not intersect $D_1$, a contradiction.
\end{proof}
By convexity, since $D_v$ contains both points $(t/4, t/4)$ and $(T, T)$, it contains all points $(a, a)$ with $t/4 \leq a \leq T$. Finally, for any $w \in N[v_5]$, there exists $a \in [t/4, T]$ such that $(a, a) \in D_v$. This result is immediate if $x_w \in [t/4, T]$ or $y_w \in [t/4, T]$. The remaining cases are:
\begin{itemize}
    \item If $x_w > T$ and $y_w > T$, then $D_w$ must intersect $(T, T)$, otherwise by Lemma~\ref{lemma:Convexity} it does not intersect any point in $[0, T]^2$, and thus does not intersect $D_5$.
    \item If $x_w < t/4$ and $y_w < t/4$, we show that $w$ cannot be adjacent to $v_5$. Indeed, $\max(x_w, y_w) + r_w \leq (1 + \sqrt{2}) \max(x_w, y_w) \leq \frac{1 + \sqrt{2}}{4} t < t$. Thus, $D_w \cap \Delta \subset [0, t)^2$, so $D_w$ does not intersect $D_5$, a contradiction.
\end{itemize}

In conclusion, we have shown that $N[v_5] \subseteq N[v]$, implying that $v$ contains $v_5$.
\end{proof}

\begin{proof}[Proof of Theorem~\ref{thm:ApproxDeltaDisks}]
Let $G$ be a $\Delta$-disk graph with a fixed disk representation. We call a vertex \emph{internal} if it contains two non-adjacent vertices, and \emph{external} otherwise. We construct a partition $V_1, \ldots, V_k$ such that for any $1 \leqslant i \leqslant k$, $V_i$ is exactly the set of external vertices of $G \setminus \bigcup_{j < i} V_j$. 

\begin{claim}
$\cs(G) \geqslant k$.
\end{claim}
\begin{proof}
We show that $G$ contains an induced $BC(k)$. We proceed by induction on $i$, proving that for any vertex $v \in V_i$, there exists a $BC(i)$ in $N\left[\{v\} \cup \bigcup_{j < i} V_j\right]$. The result is clear for $i=1$. Suppose it holds for $i-1 \geqslant 1$. Let $v \in V_i$, so $v$ contains two independent vertices $v_1$ and $v_2$ in $V_{i-1}$. Let $u_1$ (resp. $u_2$) be any vertex contained by $v_1$ (resp. $v_2$). Note that $u_1$ is not adjacent to $u_2$, because otherwise $u_2 \in N[u_1]$, implying $u_2 \in N[v_1]$ and reciprocally $v_1 \in N[u_2]$. Since $v_2$ contains $u_2$, $N[u_2] \subseteq N[v_2]$, leading to $v_1 \in N[v_2]$, which is a contradiction. By the induction hypothesis, both $v_1$ and $v_2$ contain a $BC(i-1)$, and both of them are non-adjacent, by the previous argument. Therefore, $v$ contains a $BC(i)$. 
\end{proof}

Now let us show that $\cs(G[V_\ell])$ is bounded by a constant for every $\ell \in \{1, \ldots, k\}$. Let $H = G[V_\ell]$, and construct a maximal independent set $S$ of $H$ as follows: start with the vertex of $H$ whose disk has the smallest radius in the representation (breaking ties arbitrarily), add it to $S$, remove its neighborhood, and repeat recursively. Let $v_1, \ldots, v_m$ be the obtained vertices, sorted by increasing radius $r_1, \ldots, r_m$. Observe that this total order on $S$ also corresponds to the corresponding co-comparability order.

For all $1 \leqslant i \leqslant m$, let $U_i$ be the set of vertices $v$ in $V(H)$ such that $i = \min \{j \mid vv_j \in E(H)\}$. 

\begin{claim}
For all $i \in \{1, \ldots, m\}$, $H[U_i]$ can be partitioned into at most $6$ cliques.
\end{claim}

\begin{proof}
Let $t \in U_i \setminus \{v_i\}$. Since $t$ is not adjacent to any of $v_1, \ldots, v_{i-1}$ and was not chosen to be part of $S$, it follows that $r_t \geqslant r_i$. By a classical geometrical argument, $D_t$ cannot intersect $7$ pairwise independent disks of radius at least $r_t$. It follows that $H[U_i]$ can be partitioned into $6$ cliques.
\end{proof}

From Lemma~\ref{lemma:ContainsDeltaDisk}, note that no vertex $v \in V(H)$ is adjacent to $10$ consecutive vertices of $S$, since otherwise it would contain two independent vertices. Thus, by assigning colors to the cliques of $U_i$ as $(i \mod 9, j)$ for $1 \leqslant j \leqslant 6$, we obtain a subcoloring of $H$ that uses only a constant number of colors (namely, $54$). 

In summary, we decomposed the vertex set $V(G)$ into $k$ subsets $V_1, \dots, V_k$ such that $\chi_s(G) \geq k$, and each induced subgraph $G[V_\ell]$ ($1 \leq \ell \leq k$) admits a subcoloring with a constant number of colors, $c$. This decomposition results in a $ck$-subcoloring of $G$, yielding a constant-ratio approximation algorithm.
\end{proof}

The proof of Theorem~\ref{thm:ApproxDiskGraph} follows directly from Theorem~\ref{thm:DecompositionDisk}. Specifically, any $n$-vertex disk graph can be partitioned into $O(\log^2 n)$ subgraphs, each of which is a disjoint union of $\Delta$-disk graphs. By applying the constant-factor approximation algorithm from Theorem~\ref{thm:ApproxDeltaDisks} to each of these $\Delta$-disk subgraphs, we obtain an overall $O(\log^2 n)$-approximation algorithm for \textsc{Subcoloring} on disk graphs.

\section{Conclusion and Open Questions}\label{sec:conclusion}

In this paper, we explored the subchromatic number of (unit) disk graphs, presenting both positive results and hardness results. Several intriguing questions remain open in this area of research, including two that were initially posed by Broersma et al.:

\begin{itemize}
    \item What is the complexity of \textsc{$k$-Subcoloring} on interval graphs when $k$ is part of the input?
    \item What is the complexity of \textsc{$2$-Subcoloring} on co-comparability graphs?
    \item Is there a $c$-approximation algorithm for \textsc{Subcoloring} on unit disk graphs with $c < 3$?
    \item What is the complexity of \textsc{$3$-Subcoloring} of unit disk graphs?
    \item Let $f_s(n)$ denote the maximum subchromatic number of an $n$-vertex disk graph. We have established that $f_s(n)= O(\log^3 n)$ and $f_s(n) = \Omega(\log n)$. Can we narrow this gap?
    \item Can we develop an $f(\text{OPT})$-approximation algorithm for \textsc{Subcoloring} of disk graphs, where $\text{OPT}$ is the subchromatic number of the input disk graph?
    \item What is the complexity of \textsc{$k$-Subcoloring} on $\Delta$-disk graphs?
\end{itemize}

\newpage
\bibstyle{plainurl}
\bibliography{biblio}

\newpage 
\appendix

\end{document}